\documentclass[aos,MSNbibl,seceqn,dvips]{arximspdf}

\doi{10.1214/14-AOS1204} 
\volume{42}
\issue{2}
\pubyear{2014}
\firstpage{757}
\lastpage{788}

\makeatletter
\def\mid{\vert}
\newproclaim{cond}{Condition}
\newcommand{\rright}{\right}
\newcommand{\lleft}{\left}
\newcommand{\rrVert}{\Vert}
\newcommand{\llVert}{\Vert}
\def\implies{\Longrightarrow}
\newtheorem{theorem}{Theorem}
\newtheorem{corollary}{Corollary}
\newtheorem{lemma}{Lemma}
\newproclaim{example}{Example}
\newproclaim{remark}{Remark}
\newtheorem{algorithm}{Algorithm}

\renewcommand{\hat}{\widehat}

\newcommand{\BB}{\mathbb{B}}

\newcommand{\cc}{\bar{c}}
\newcommand{\Gn}{\mathbb{G}_n}
\newcommand{\En}{\mathbb{E}_n}
\newcommand{\Ep}{\mathrm{E}}
\newcommand{\sign}{\operatorname{sign}}
\def\RR{{\mathbb{R}}}

\def\supp{\operatorname{supp}}
\def\kk{{\bar\kappa}}
\def\lkk{{\underline{\kappa}}}
\def\bbb{{\bar\zeta}}
\def\BB{B_n}
\def\barEp{\bar{\mathbf{E}}}
\newcommand{\LASSO}{\sqrt{\mathrm{Lasso}}}
\def\diag{\operatorname{diag}}
\def\ii{{}}

\makeatother

\begin{document}
\begin{frontmatter}

\title{Pivotal estimation via square-root Lasso in nonparametric
regression}
\runtitle{$\sqrt{\mathrm{LASSO}}$ for nonparametric regression}

\begin{aug}
\author[a]{\fnms{Alexandre} \snm{Belloni}\corref{}\ead[label=e1]{abn5@duke.edu}},
\author[b]{\fnms{Victor} \snm{Chernozhukov}\ead[label=e2]{vchern@mit.edu}}
\and
\author[c]{\fnms{Lie} \snm{Wang}\ead[label=e3]{lwang@mit.edu}}
\runauthor{A. Belloni, V. Chernozhukov and L. Wang}
\affiliation{Duke University, Massachusetts Institute of Technology\break and
Massachusetts Institute of Technology}
\address[a]{A. Belloni\\
Fuqua School of Business\\
Duke University\\
100 Fuqua Drive\\
Durham, North Carolina 27708\\
USA\\
\printead{e1}} 
\address[b]{V. Chernozhukov\\
Department of Economics\\
Massachusetts Institute of Technology\\
52 Memorial Drive\\
Cambridge, Massachusetts 02142\\
USA\\
\printead{e2}}
\address[c]{L. Wang\\
Department of Mathematics\\
Massachusetts Institute of Technology\\
77 Massachusetts Avenue\\
Cambridge, Massachusetts 02139\\
USA\\
\printead{e3}}
\end{aug}

\received{\smonth{9} \syear{2012}}
\revised{\smonth{12} \syear{2013}}

%
\begin{abstract}
We propose a self-tuning $\LASSO$ method that simultaneously resolves
three important practical problems in high-dimensional regression
analysis, namely it handles the unknown scale, heteroscedasticity and
(drastic) non-Gaussianity of the noise. In addition, our analysis
allows for badly behaved designs, for example, perfectly collinear
regressors, and generates
sharp bounds even in extreme cases, such as the infinite variance case
and the noiseless case, in contrast to Lasso. We establish various
nonasymptotic bounds for $\LASSO$ including prediction norm rate and sparsity.
Our analysis is based on new impact factors that are tailored for
bounding prediction norm. In order to cover heteroscedastic
non-Gaussian noise, we rely on moderate deviation theory for
self-normalized sums to achieve Gaussian-like results under weak
conditions. Moreover, we derive bounds on the performance of ordinary
least square (ols) applied to the model selected by $\LASSO$ accounting
for possible misspecification of the selected model. Under mild
conditions, the rate of convergence of ols post $\LASSO$ is as good as
$\LASSO$'s rate. As an application, we consider the use of $\LASSO$ and
ols post $\LASSO$ as estimators of nuisance parameters in a generic
semiparametric problem (nonlinear moment condition or $Z$-problem),
resulting in a construction of $\sqrt{n}$-consistent and asymptotically
normal estimators of the main parameters.
\end{abstract}

%


\end{frontmatter}

\section{Introduction}\label{Sec:Intro}

We consider a nonparametric regression model:
%
%
\begin{equation}
y_i = f(z_i) + \sigma
\varepsilon_i,\qquad i=1,\ldots,n,
\end{equation}
where $y_i$'s are the outcomes, $z_i$'s are vectors of fixed basic
covariates, $\varepsilon_i$'s are independent noise, $f$ is the regression
function and $\sigma$ is an unknown scaling parameter. The goal is to recover
the values $(f_i)_{i=1}^n = (f(z_i))_{i=1}^n$ of the regression
function $f$ at $z_i$'s. To achieve this goal, we use linear
combinations of technical regressors $x_i=P(z_i)$ to approximate $f$,
where $P(z_i)$ is a dictionary of
$p$-vector of transformations of $z_i$. We are interested in the
high dimension low sample size case, where we potentially use
$p > n$, to obtain a flexible approximation. In particular, we are
interested in cases
where the regression function can be well approximated by a sparse linear
function of~$x_i$.\looseness=-1

The model above can be written as
$y_i = x_i' \beta_0 + r_i + \sigma\varepsilon_i$, where $f_i =
f(z_i)$ and
$r_i:= f_i - x_i' \beta_0$ is the
approximation error. The vector $\beta_0$ is defined as a solution of
an optimization problem to compute the oracle risk, which balances bias
and variance (see Section~\ref{Sec:Model}). The cardinality of the
support of $\beta_0$ is denoted by $s:= \| \beta_0\|_0$.
It is well known that ordinary least squares (ols) is generally
inconsistent when $p>n$. However, the sparsity assumption, namely that
$s \ll n$, makes it
possible to estimate these models effectively by searching for
approximately the right set of the regressors. In particular,
$\ell_1$-penalization has played a central
role \cite
{BickelRitovTsybakov2009,BVDG2011,CandesTao2007,Koltchinskii2009,MY2007,vdGeer,ZhangHuang2006,Wainright2006}.
It was demonstrated that $\ell_1$-penalized
least squares estimators can achieve the rate $\sigma\sqrt{s/n}
\sqrt{\log p}$, which is very close to the oracle rate
$\sigma\sqrt{s/n}$ achievable when the true model is known.
Importantly, in the context of linear regression, these $\ell
_1$-regularized problems can be cast as convex optimization problems
which make them computationally
efficient (computable in polynomial time). We
refer to \cite
{BickelRitovTsybakov2009,BVDG2011,BuneaTsybakovWegkamp2007,BuneaTsybakovWegkamp2007b,FanLv2006,Lounici2008,LouniciPontilTsybakovvandeGeer2009,RosenbaumTsybakov2008,vdGeer}
for a more detailed review of the existing literature which has
focused on the homoscedastic case.\looseness=-1

In this paper, we attack the problem of nonparametric regression under
non-Gaussian, heteroscedastic errors $\varepsilon_i$, having an unknown
scale $\sigma$. We propose to use a self-tuning $\LASSO$ which is
pivotal with respect to the scaling parameter~$\sigma$, and which
handles non-Gaussianity
and heteroscedasticity in the errors. The resulting rates and
performance guarantees are very similar to the Gaussian case, due to
the use of self-normalized moderate deviation theory. Such results and
properties,\footnote{Earlier literature, for example, in bounded
designs \cite{BVDG2011}, provides bounds using refinements of
Nemirovski's inequality; see \cite{Duembgen2010}. These results provide
rates as good as in the Gaussian case. However, when the design is
unbounded (e.g., regressors generated as realizations of independent
Gaussian random variables), the rates of convergence provided by these
techniques are no longer sharp. The use of self-normalized moderate
deviations in the present context allows to handle the latter cases,
with sharp rates.} particularly the pivotality with respect to the
scale, are in contrast to the previous results and methods on others
$\ell_1$-regularized methods, for example, Lasso and Dantzig selector
that use penalty levels that depend linearly on the unknown scaling
parameter $\sigma$.

There is now a growing literature on high-dimensional linear
models\footnote{There is also a literature on penalized median
regression, which can be used in the case of symmetric errors,
since these methods are independent of the unknown $\sigma$, cf.
\cite{BC-SparseQR,Wang2012}.} allowing for unknown scale $\sigma$.
St\"adler et al. \cite{StadlerBuhlmanGeer} propose a $\ell_1$-penalized maximum
likelihood estimator for parametric Gaussian regression models.
Belloni et al. \cite{BCW-SqLASSO} consider $\LASSO$ for a parametric
homoscedastic model with both Gaussian and non-Gaussian errors and
establish that the choice of the penalty parameter in $\LASSO$
becomes pivotal with respect to $\sigma$. van de Geer
\cite{vandeGeerBuhlmannRitov2013} considers an equivalent formulation of the
(homoscedastic)
$\LASSO$ to establish finite sample results and derives
results in the parametric homoscedastic Gaussian setting.
Chen and Dalalyan~\cite{ChenDalalyan2012} consider scaled fused Dantzig selector to
allow for different sparsity patterns and provide
results under homoscedastic Gaussian errors. Belloni and Chernozhukov~\cite{BC-PostLASSO} study Lasso with a plug-in estimator of the noise level based on
Lasso iterations
in a parametric homoscedastic setting. Chr{\'e}tien and Darses
\cite{ChretienDarses2012} study plug-in estimators and a trade-off penalty choice between
fit and penalty in the parametric case with homoscedastic Gaussian
errors under random support assumption (similar to
\cite{CandesPlan2009}) using coherence condition. In a trace
regression model for recovery of a matrix, \cite{Klopp2011} proposes and
analyses a version of the $\LASSO$ under
homoscedasticity. A comprehensive review is
given in \cite{GiraudHuetVerzelen2012}. All these works rely
essentially on the restricted eigenvalue condition
\cite{BickelRitovTsybakov2009} and homoscedasticity and do not
differentiate penalty levels across components.


In order to address the nonparametric, heteroscedastic and non-Gaussian
cases, we develop covariate-specific penalty loadings.
To derive a practical and theoretically justified choice of penalty
level and loadings, we need to account for the impact of the
approximation error. We rely on moderate deviation
theory for self-normalized sums of \cite{jing:etal} to achieve
Gaussian-like results in many
non-Gaussian cases provided $\log p = o(n^{1/3})$, improving upon
results derived in the parametric case that required $\log p
\lesssim\log n$, see \cite{BCW-SqLASSO}. (In the context of standard
Lasso, the self-normalized moderate deviation theory was first employed
in \cite{BellChenChernHans:nonGauss}.) 

Our first contribution is the proposal of new design and noise impact
factors, in order to allow for more general designs. Unlike previous
conditions, these factors are tailored for establishing performance
bounds with respect to the prediction norm, which is appealing in
nonparametric problems. In particular, collinear designs motivate our
new condition. In studying their properties, we further exploit the
oracle based definition of the approximating function. The analysis
based on these impact factors complements the analysis based on
restricted eigenvalue proposed in \cite{BickelRitovTsybakov2009} and
compatibility condition in \cite{GeerBuhlmann2009}, which are more
suitable for establishing rates for $\ell_k$-norms. 

The second contribution is a set of finite sample upper bounds and lower
bounds for estimation errors under prediction norm, and upper bounds on the
sparsity of the $\LASSO$ estimator. These results are ``geometric,''
in that
they hold conditional on the design and errors provided some key events
occur. We further develop primitive sufficient conditions that allow for
these results to be applied to heteroscedastic non-Gaussian errors. We
also give
results for other norms in the supplementary material \cite{BCWsupp}.

The third contribution develops properties of the estimator that
applies ordinary least squares (ols)\vspace*{1pt} to the model selected by $\LASSO$.
Our focus is on the case that $\LASSO$ fails to achieve perfect model
selection, including cases where the oracle model is not completely
selected by $\LASSO$. This is usually the case in a nonparametric
setting. This estimator intends to remove the potentially significant
bias toward zero introduced by the $\ell_1$-norm regularization
employed in the $\LASSO$ estimator.

The fourth contribution is to study two extreme cases: (i)
parametric noiseless case and (ii) nonparametric infinite
variance case. $\LASSO$ has interesting theoretical
properties for these two extreme cases. For case (i), $\LASSO$
can achieves exact recovery in sharp contrast to Lasso. For case (ii),
$\LASSO$ estimator can
still be consistent with penalty choice that does not depend on the
scale of the noise. We develop the necessary modifications of
the penalty loadings and derive finite-sample bounds for the case of
symmetric noise. When noise
is Student's $t$-distribution with $2$ degrees of freedom, we recover
Gaussian-noise
rates up to a multiplicative factor of $\log^{1/2} n$.

The final contribution is to provide an application of $\LASSO$ methods
to a generic semiparametric problem, where
some low-dimensional parameters are of interest and $\LASSO$ methods
are used to estimate nonparametric nuisance parameters. These
results extend the $\sqrt{n}$ consistency and
asymptotic normality results of \cite
{BellChernHans:Gauss,BellChenChernHans:nonGauss}
on a rather specific linear model to a generic nonlinear problem, which
covers smooth frameworks in statistics and in econometrics, where
the main parameters of interest are defined via nonlinear instrumental
variable/moment conditions
or $Z$-conditions containing unknown nuisance functions (as in \cite
{chamberlain}). This and all the above results illustrate the wide
applicability of the proposed estimation procedure.

\textit{Notation}. To make asymptotic statements, we assume that $n
\to
\infty$ and $p=p_n \to\infty$, and we allow for $s=s_n \to
\infty$. In what follows, all parameters are indexed by the sample size
$n$, but we omit the index whenever it does not cause confusion. We
work with i.n.i.d., independent but not necessarily identically
distributed data, $(w_i)_{i=1}^n$, with $k$-dimensional real vectors $w_i$
containing $y_i \in\mathbb{R}$ and $z_i \in\mathbb{R}^{p_z}$, the
latter taking values in a set $\mathcal{Z}$. We use the notation $(a)_+
= \max\{a,0\}$, $a \vee b = \max\{ a, b\}$ and $a \wedge b = \min\{
a,
b \}$. The $\ell_2$-norm is denoted by $\|\cdot\|$, the $\ell_1$-norm
is denoted by $\|\cdot\|_1$, the $\ell_\infty$-norm is denoted by $\|
\cdot\|_\infty$,
and the $\ell_0$-``norm'' $\|\cdot\|_0$ denotes the number of nonzero
components of a vector. The transpose of a matrix $A$ is denoted by
$A'$. Given a vector $\delta\in\RR^p$, and
a set of indices $T \subset\{1,\ldots,p\}$, we denote by $\delta_T$ the
vector in which $\delta_{Tj} = \delta_j$ if $j\in T$, $\delta_{Tj}=0$
if $j \notin T$, and by $|T|$ the cardinality of $T$. For a measurable
function $f\dvtx\RR^{k} \to\RR$, the symbol $\Ep[f(w_i)]$ denotes the
expected value of $f(w_i)$; $\En[f(w)]$ denotes the average $n^{-1}
\sum_{i=1}^n f(w_i)$; $\barEp[f(w)]$ denotes the average expectation
$n^{-1} \sum_{i=1}^n \Ep[f(w_i)]$; and $\Gn(f(w))$ denotes $n^{-1/2}
\sum_{i=1}^{n} (f(w_{i}) - \Ep[ f(w_{i}) ])$. We will work with
regressor values $(x_i)_{i=1}^n$ generated via $x_i = P(z_i)$, where
$P(\cdot)\dvtx\mathcal{Z} \mapsto\RR^p$ is a measurable dictionary of
transformations, where $p$ is potentially larger than $n$. We define
the prediction norm of a vector $\delta\in\RR^p$ as $\|\delta\|_{2,n}
= \{\En[(x_\ii'\delta)^2]\}^{1/2}$, and given values $y_1,\ldots,y_n$
we define $\widehat Q(\beta) = \En[(y_\ii-x_\ii'\beta)^2]$.
We use the notation $a \lesssim b$ to denote $a \leq C b$ for some
constant $C>0$ that does not depend on $n$ (and, therefore, does not
depend on quantities indexed by $n$ like $p$ or $s$); and $a\lesssim_P
b$ to denote $a=O_P(b)$. $\Phi$ denotes the cumulative distribution of
a standard Gaussian distribution and $\Phi^{-1}$ its inverse function.

\section{Setting and estimators}\label{Sec:Model}

Consider the nonparametric regression model:
%
%
\begin{eqnarray}
\label{Def:NP} y_i &=& f(z_i) + \sigma
\varepsilon_i,\qquad \varepsilon_i \sim F_i,
\nonumber
\\[-8pt]
\\[-8pt]
\nonumber
\Ep[
\varepsilon_i]& =&0,\qquad i=1,\ldots,n,\qquad \barEp \bigl[\varepsilon_\ii^2
\bigr] = 1,
\end{eqnarray}
where $z_i$ are vectors of fixed regressors, $\varepsilon_i$ are
independent errors, and $\sigma$ is the scaling factor of the errors.
In order to recover the regression function $f$, we consider linear
combinations of the covariates $x_{i}=P(z_i)$ which are $p$-vectors of
transformation of $z_i$ normalized so that $\En[x_{\ii j}^2]=1$
($j=1,\ldots,p$).

The goal is to estimate the value of the nonparametric regression
function $f$ at the design points, namely the values $(f_i)_{i=1}^n:=
(f(z_i))_{i=1}^n$. In the nonparametric settings, the regression
functions $f$ are generically nonsparse. However, often they can be
well approximated by a sparse model $x'\beta_0$. One way to find such
approximating model is to let $\beta_0$ be a solution of the following
risk minimization
problem:
%
%
\begin{equation}
\label{oracle} \min_{\beta\in\RR^p} \En \bigl[ \bigl(f_\ii-
x_\ii'\beta \bigr)^2 \bigr] +
\frac
{\sigma^2\|
\beta\|_0}{n}.
\end{equation}
The problem (\ref{oracle}) yields the so called oracle risk---an upper
bound on the risk of the best $k$-sparse
least squares estimator in the case of homoscedastic Gaussian errors,
that is, the best estimator among all least squares estimators that use
$k$ out of $p$ components of $x_i$ to estimate $f_i$. The solution
$\beta_0$ achieves a balance between the mean square of the
approximation error $r_i:=f_i-x_i'\beta_0$ and the variance, where the
latter is determined by the complexity $\| \beta_0\|_0$ of the model
(number of nonzero components of $\beta_0$).

In what follows, we call $\beta_0$ the target parameter value, $T:=
\supp(\beta_0)$ the oracle model, $s:= |T| = \| \beta_0\|_0$ the
dimension of the oracle model, and $x_i'\beta_0$ the oracle or the best
sparse approximation to $f_i$. We note that $T$ is generally unknown.
We summarize the preceding discussion as follows.

\renewcommand{\thecond}{ASM}
\begin{cond}\label{condASM}
We have data $\{(y_i,z_i) \dvtx
i=1,\ldots,n\}$ that for each $n$ obey the regression model (\ref
{Def:NP}), where $y_i$ are the outcomes, $z_i$ are vectors of fixed
basic covariates, the regressors $x_i:=P(z_i)$ are transformations of
$z_i$, and $\varepsilon_i$ are i.n.i.d. errors. The vector
$\beta_0$ is defined by (\ref{oracle}) where\vadjust{\goodbreak} the regressors $x_i$ are
normalized so that $\En[x_{\ii j}^2]=1$, $j=1,\ldots,p$. We let
%
%
\begin{equation}
\label{Def:Many}\qquad T:= \supp(\beta_0),\qquad s:=|T|,\qquad r_i:=
f_i-x_i'\beta_0 \quad\mbox{and}\quad
c_s^2:= \En \bigl[r_\ii^2
\bigr].
\end{equation}
\end{cond}
%
%
\begin{remark}[(Targeting $x_i'\beta_0$ is the same as targeting
$f_i$'s)] We focus
on estimating the oracle model $x_i'\beta_0$ using estimators of the
form $x_i'\hat\beta$,
and we seek to bound estimation errors with respect to the prediction norm
$\|\hat\beta-\beta_0\|_{2,n}:= \{\En[(x_\ii'\beta_0-x_\ii'\hat
\beta
)^2]\}^{1/2}$.
The bounds on estimation errors for the ultimate target $f_i$ then
follow from the triangle inequality, namely
%
%
\begin{equation}
\label{Def:Triangle}\sqrt{\En \bigl[ \bigl(f_\ii-x_\ii'
\hat\beta \bigr)^2 \bigr]} \leq\|\hat\beta-\beta_0
\|_{2,n} + c_s.
\end{equation}
\end{remark}

%
\begin{remark}[(Bounds on the approximation error)]\label
{Remark:ApproximationError} The approximation errors
typically satisfy $c_s \leq K \sigma\sqrt{(s \vee1)/n}$ for some fixed
constant $K$, since the optimization problem (\ref{oracle}) balances
the (squared) norm of
the approximation error (the norm of the bias) and the variance; see
\cite{TsybakovBook,BC-LectureNotes,BC-PostLASSO}.
In particular, this condition holds for wide classes of functions; see
Example \ref{exS} of Section~\ref{cec:primitive} dealing with Sobolev classes
and Section C.2 
of supplementary material \cite{BCWsupp}.
\end{remark}

\subsection{Heteroscedastic \texorpdfstring{$\LASSO$}{square-root Lasso}}

In this section, we formally define the estimators which are tailored
to deal with heteroscedasticity.

We propose to define the $\LASSO$ estimator as
%
%
\begin{equation}
\label{Def:LASSOmod} \hat\beta\in\arg\min_{\beta\in\RR^p} \sqrt{\widehat Q(
\beta)} + \frac{\lambda}{n} \|\Gamma\beta\|_1,
\end{equation}
where $\widehat Q(\beta) = \En[(y_\ii-x_\ii'\beta)^2]$, $\Gamma
=\diag
(\gamma_1,\ldots,\gamma_p)$ is a diagonal matrix of penalty loadings.
The scaled $\ell_1$-penalty allows component specific adjustments to
more efficiently deal with heteroscedasticity.\footnote{When errors are
homoscedastic, we can set $\Gamma= I_p$. In the heteroscedastic case, using
$\Gamma=I_p$ may require setting $\lambda$ too conservatively, leading
to over-penalization and worse performance bounds. In the paper, we
develop data-dependent choice of $\Gamma$ that allows us to avoid
over-penalization thereby improving the performance.} Throughout, we
assume $\gamma_{j} \geq1$ for $j=1,\ldots,p$.

In order to reduce the shrinkage\vspace*{1pt} bias of $\LASSO$, we consider the post
model selection estimator that applies ordinary least squares (ols) to
a model $\widehat T$ that contains the model selected by $\LASSO$.
Formally, let $\widehat T$ be such that
\[
\supp( \hat\beta) = \bigl\{ j \in\{1,\ldots,p\} \dvtx|\hat\beta_j|
> 0 \bigr\} \subseteq\widehat T,
\]
and define the ols post $\LASSO$ estimator $\widetilde\beta$
associated with $\widehat T$ as
%
%
\begin{equation}
\label{Def:TwoStep} \widetilde\beta\in\arg\min_{\beta\in\mathbb
{R}^p} \sqrt{\widehat
Q(\beta)} \dvtx\beta_j = 0\qquad \mbox{if } j \notin\widehat T.
\end{equation}
A sensible choice for $\widehat T$ is simply to set $\widehat T=\supp(
\hat\beta)$. Moreover, we allow for additional components
(potentially selected through an arbitrary data-dependent procedure) to
be added, which is relevant for practice.

\subsection{Typical conditions on the Gram matrix}\label{Sec:CondMatrix}
The Gram matrix $\En[x_\ii x_\ii']$ plays an important role in
the analysis of estimators in this setup. When $p>n$, the
smallest eigenvalue of the Gram matrix is $0$, which creates
identification problems. Thus, to restore identification, one
needs to restrict the type of deviation vectors $\delta$
corresponding to the potential deviations of the estimator from
the target value $\beta_0$. Because of the $\ell_1$-norm
regularization, the following restricted set is important:
\[
\Delta_{\cc} = \bigl\{\delta\in\RR^p\dvtx\|\Gamma
\delta_{T^c}\|_{1} \leq\cc\|\Gamma\delta_{T}
\|_{1}, \delta\neq0 \bigr\}\qquad \mbox{for } \cc\geq1.
\]
The restricted eigenvalue
$\kappa_{\cc}$ of the Gram matrix $\En[x_\ii x_\ii']$ is defined
as
%
%
\begin{equation}
\label{RE} \kappa_{\cc}: = \min_{\delta\in\Delta_{\cc}}
\frac{\sqrt{s}\|
\delta
\|_{2,n}}{\|\Gamma\delta_T\|_1 }.
\end{equation}
The restricted eigenvalues can depend on $n$, $T$, and $\Gamma$, but we
suppress the dependence in our notation. The restricted
eigenvalues (\ref{RE}) are variants of the restricted eigenvalue
introduced in \cite{BickelRitovTsybakov2009} and of the compatibility
condition in \cite{GeerBuhlmann2009}
that accommodate the penalty loadings $\Gamma$. They were proven to be
useful for many designs of interest
specially for establishing $\ell_k$-norm rates. Below we suggest their
generalizations that are useful for deriving rates in prediction norm.

The minimal and maximal $m$-sparse eigenvalues of a matrix $M$,
%
%
\begin{eqnarray}
\label{SEdef} \phi_{\mathrm{min}}(m,M)&:=& \min_{\|\delta_{T^c}\|
_{0} \leq m,
\delta\neq0
}
\frac{ \delta'M\delta}{\|\delta\|^2},
\nonumber
\\[-8pt]
\\[-8pt]
\nonumber
\phi_{\mathrm{max}}(m,M)&:=& \max_{\|
\delta_{T^c}\|_{0} \leq m, \delta\neq0
}
\frac{ \delta'M\delta}{\|\delta\|^2}.
\end{eqnarray}
Typically, we consider $M=\En[x_\ii x_\ii']$ or $M=\Gamma^{-1}\En
[x_\ii x_\ii']\Gamma^{-1}$. When $M$ is not specified, we mean $M=\En
[x_\ii x_\ii']$, that is, $\phi_{\mathrm{min}}(m)=\phi_{\mathrm
{min}}(m,\En[x_\ii x_\ii'])$ and
$\phi_{\mathrm{max}}(m)=\phi_{\mathrm{max}}(m,\En[x_\ii x_\ii'])$.
These quantities play an important role in the sparsity and post model
selection analysis. Moreover, sparse
eigenvalues provide a simple sufficient condition to bound
restricted eigenvalues; see \cite{BickelRitovTsybakov2009}.

\section{Finite-sample analysis of \texorpdfstring{$\LASSO$}{square-root Lasso}}\label{sec:finite-result}

Next, we establish several finite-sample results regarding the
$\LASSO$ estimator. Importantly, these results are based on new impact
factors that can be very well behaved under repeated (i.e., collinear)
regressors, and which strictly generalize the restricted eigenvalue
(\ref{RE}) and compatibility constants.

The following event plays a central role in the analysis:
%
%
\begin{equation}
\label{RegEvent}\quad \lambda/n \geq c \bigl\|\Gamma^{-1}\widetilde S
\bigr\|_\infty \qquad\mbox{where } \widetilde S:= \En \bigl[x_\ii(\sigma
\varepsilon_\ii+r_\ii) \bigr]/\sqrt{\En \bigl[(\sigma
\varepsilon_\ii+ r_\ii)^2 \bigr]}
\end{equation}
%
is the score of $\widehat Q^{1/2}$ at $\beta_0$ ($\widetilde S = 0$ if
$\En[(\sigma\varepsilon_\ii+ r_\ii)^2]=0$). Throughout the
section, we
assume such event holds. Later we provide choices of
$\lambda$ and $\Gamma$ based on primitive conditions such that the
event in (\ref{RegEvent}) holds with a high probability. %


\subsection{New noise and design impact factors}\label{Sec:NewIden}
We define the following \textit{noise} and \textit{design} impact
factors for a constant $c>1$:
%
%
\begin{eqnarray}
\label{Def:rho} \varrho_{c} &:= & \sup_{\tiny\|\delta\|_{2,n}>0,
\delta\in R_c }
\frac{|\widetilde S'\delta|}{ \|\delta\|_{2,n}},
\\
\label{Def:kk} \kk&:= & \inf_{\|\Gamma(\beta_0 + \delta) \|_{1}
< \|
\Gamma\beta_0\|_{1} } \frac{\sqrt{s}\|\delta\|_{2,n}}{ \|\Gamma
\beta
_0\|_{1} - \|\Gamma(\beta_0 + \delta) \|_{1}},
\end{eqnarray}
where $R_c:=\{ \delta\in\RR^p\dvtx\|\Gamma\delta\|_{1} \geq c
( \|
\Gamma(\beta_0 + \delta) \|_{1} - \|\Gamma\beta_0\|_{1} )\}$. For the
case $\beta_0=0$, we define $\varrho_c = 0$ and $\bar\kappa= \infty$.
These quantities depend on $n$, $\beta_0$ and $\Gamma$, albeit we
suppress this when convenient.

An analysis based on the quantities $\varrho_{c}$ and $\kk$ will be
more general than the one relying only on restricted eigenvalues (\ref
{RE}). This follows because
(\ref{RE}) yields one possible way to bound both $\kk$ and $\varrho
_{c}$, namely,
\begin{eqnarray*}
\kk\geq\lkk&:= &\inf_{{\delta\in{\rm int}(\Delta_1)}}
\frac{\sqrt{s}\|\delta\|_{2,n}}{ \|\Gamma\delta_T\|_{1 }-\|\Gamma
\delta
_{T^c}\|_{1
}} \geq\min_{\delta\in\Delta_{1}} \frac{\sqrt{s}\|\delta\|
_{2,n}}{\|
\Gamma\delta_T\|_1} \geq\min
_{\delta\in\Delta_{\cc}} \frac
{\sqrt{s}\|\delta\|_{2,n}}{\|\Gamma\delta_T\|_1} = \kappa_{\cc},
\\
\varrho_{c} &\leq&\sup_{\delta\in\Delta_{\cc}}\frac{\|\Gamma
^{-1}\widetilde S\|_{\infty}\|\Gamma\delta\|_1}{\|\delta\|_{2,n}} \leq
\sup_{\delta\in\Delta_{\cc}} \frac{\|\Gamma^{-1}\widetilde S\|
_{\infty
}(1+\cc)\|\Gamma\delta_T\|_{1}}{\|\delta\|_{2,n}}\\
& \leq&\frac
{(1+\cc
)\sqrt{s}}{\kappa_{\cc}}\|
\Gamma^{-1}\widetilde S\|_{\infty},
\end{eqnarray*}
where $c>1$ and $\cc:=(c+1)/(c-1)>1$. The quantities $\kk$ and
$\varrho
_{c}$ can be well behaved (i.e., $\kk>0$ and $\varrho_{c} < \infty$)
even in the presence of repeated (i.e., collinear) regressors (see
Remark~\ref{re4} for a simple example), while restricted eigenvalues and
compatibility constants would be zero in that case.

The design impact factor $\kk$ in (\ref{Def:kk}) strictly generalizes
the original restricted eigenvalue (\ref{RE})
proposed in \cite{BickelRitovTsybakov2009} and the compatibility
constants proposed in
\cite{GeerBuhlmann2009} and in \cite{Geer2007}.\hskip.2pt\footnote{The compatibility
condition defined in \cite{Geer2007} is defined as: $\exists\nu(T)>0$
such that
\[
\inf_{\delta\in\Delta_3}\frac{\sqrt{s}\|\delta\|_{2,n}}{(1+\nu
(T))\|
\Gamma\delta_T\|_1-\|\Gamma\delta_{T^c}\|_1}>0.
\]
We
have that $\kk\geq\lkk$, where $\lkk$ corresponds to setting $\nu
(T)=0$ and using $\Delta_1$ in place of $\Delta_3$, which
strictly weakens \cite{Geer2007}'s definition. Allowing for $\nu(T)=0$
is necessary for allowing collinear regressors.}
The design conditions based on these concepts are relatively weak, and
hence (\ref{Def:kk}) is a useful concept.

The noise impact factor $\varrho_{c}$ also plays an important role. It
depends on the noise, design and
approximation errors, and can be controlled via empirical process
methods. Note that under (\ref{RegEvent}),
the deviation $\hat\delta=\hat\beta-\beta_0$ of the $\LASSO$ estimator
from $\beta_0$ obeys $\hat\delta\in R_c$, explaining its appearance
in the definition of $\varrho_c$. The lemmas below summarize the above
discussion.

%
\begin{lemma}[(Bounds on and invariance of design impact factor)]\label
{Lemma:NewKappa}
Under Condition \ref{condASM}, we have $\kk\geq\lkk\geq\kappa_1\geq\kappa
_{\cc}$. Moreover, if copies of regressors are included with the same
corresponding penalty loadings, the lower bound $\lkk$ on $\kk$ does
not change.
\end{lemma}

%
\begin{lemma}[(Bounds on and invariance of noise impact factor)]\label
{Lemma:NewRho}
Under Condition \ref{condASM}, we have $\varrho_{c} \leq(1+\cc) \sqrt{s} \|
\Gamma
^{-1}\widetilde S\|_\infty/ \kappa_{\cc}$. Moreover, if copies of
regressors with indices $j \in T^c$ are included with the same
corresponding penalty loadings,
$\varrho_{c}$ does not change (see also Remark~\ref{re4}).
\end{lemma}


%
\begin{lemma}[(Estimators belong to restricted sets)]\label{Lemma:Domination}
Assume that for some $c > 1$ we have $\lambda/n \geq c \|\Gamma
^{-1}\widetilde S\|_{\infty}$, then $\hat\delta\in R_c$.
The latter condition implies that $\hat\delta\in\Delta_{\bar c}$ for
$\bar c = (c+1)/(c-1)$.
\end{lemma}



\subsection{Finite-sample bounds on \texorpdfstring{$\LASSO$}{square-root Lasso}}\label{Sec:Prediction}


In this section, we derive finite-sample bounds for the prediction norm
of the $\LASSO$ estimator. These bounds are established under
heteroscedasticity, without knowledge of the scaling parameter $\sigma
$, and using the impact factors proposed in Section~\ref{Sec:NewIden}.
For $c>1$, let $\cc=(c+1)/(c-1)$ and consider the conditions
%
%
\begin{equation}
\label{MainFSCondition} \lambda/n \geq c \bigl\|\Gamma^{-1}\widetilde S
\bigr\|_\infty\quad\mbox{and}\quad \bbb:= \lambda\sqrt{s}/(n\kk)<1.
\end{equation}

%
\begin{theorem}[(Finite sample bounds on estimation error)]\label
{Thm:NewRateSquareRootLASSONonparametric}
Under Condition \ref{condASM} and (\ref{MainFSCondition}), we have
\[
\|\hat\beta- \beta_0\|_{2,n} \leq2\sqrt{\hat Q(
\beta_0)} \BB,\qquad \BB:= \frac{\varrho_{c} + \bbb}{1-\bbb^2}.
\]
\end{theorem}

We recall that the
choice of $\lambda$ does not depend on the scaling parameter
$\sigma$. The impact of $\sigma$ in the bound of Theorem~\ref
{Thm:NewRateSquareRootLASSONonparametric} comes through
the factor
$\hat Q^{1/2}(\beta_0) \leq\sigma\sqrt{\En[\varepsilon_\ii^2]} + c_s$
where $c_s$ is the size of the approximation error defined in
Condition~\ref{condASM}.
Moreover, under typical conditions that imply $\kappa_{\cc}$ to be
bounded away from zero, for example,
under Condition \ref{condP} of Section~\ref{cec:primitive} and standard choice of
penalty, we have with a high probability
\[
\BB\lesssim\sqrt{\frac{ s \log(p \vee n) }{n}}\quad \implies\quad\|\hat \beta- \beta_0
\|_{2,n} \lesssim\sigma\sqrt{\frac{ s \log(p \vee
n) }{n}}.
\]
%
Thus, Theorem~\ref{Thm:NewRateSquareRootLASSONonparametric} generally
leads to the same rate of convergence as in the
case of the Lasso estimator that knows $\sigma$ since
$\En[\varepsilon_\ii^2]$ concentrates around 1 under (\ref{Def:NP})
and provided a law of large numbers holds. We derive performance bounds
for other
norms of interest in the supplementary material~\cite{BCWsupp}.\vspace*{1pt}


%
%
%

The next result deals with $\widehat
Q(\hat\beta)$ as an estimator for $\widehat Q(\beta_0)$ and $\sigma^2$.

%
\begin{theorem}[(Estimation of $\sigma$)]\label{Thm:Q}
Under Condition \ref{condASM} and (\ref{MainFSCondition})
\[
-2\varrho_{c}\sqrt{\widehat Q(
\beta_0)}\BB\leq\sqrt{\widehat Q(\widehat\beta)} - \sqrt {\widehat
Q(\beta_0)}  \leq2\bbb\sqrt{\widehat Q(\beta_0)}\BB.
\]
Under only Condition \ref{condASM}, we have
\[
\bigl\vert\sqrt{\widehat Q(\widehat\beta)}- \sigma
\bigr\vert \leq\| \widehat\beta-\beta_0\|_{2,n} +
c_s + \sigma\bigl|\En \bigl[\varepsilon_\ii^2
\bigr] - 1\bigr|.
\]
\end{theorem}
We note that further bounds on $|\En[\varepsilon_\ii^2] - 1|$ are implied
by von Bahr--Esseen's and Markov's inequalities, or by self-normalized
moderate deviation (SNMD) theory as in Lemma~\ref{Lemma: MDSN}. As a
result, the theorem implies consistency $| \widehat Q^{1/2}(\widehat
\beta)- \sigma| =o_P(1)$ under mild moment conditions; see
Section~\ref{cec:primitive}. Theorem~\ref{Thm:Q} is also useful for
establishing
the following sparsity properties.



%
\begin{theorem}[(Sparsity bound for $\LASSO$)]\label{Thm:Sparsity}
Suppose Condition \ref{condASM}, (\ref{MainFSCondition}), $\widehat Q(\beta
_0)>0$, and $2\varrho_{c} \BB\leq1/(c\cc)$. Then we have
\[
\bigl|\supp(\widehat\beta)\bigr| \leq s \cdot4\cc^2 (\BB/\bbb
\kk)^2 \min_{m\in
\mathcal{M}}\phi_{\mathrm{max}}\bigl(m,
\Gamma^{-1}\En \bigl[x_\ii x_\ii'
\bigr] \Gamma^{-1}\bigr),
\]
where $\mathcal{M}=\{ m \in\mathbb{N}\dvtx m > s \phi_{\mathrm
{max}}(m, \Gamma ^{-1}\En[x_\ii x_\ii']\Gamma^{-1})\cdot8\cc^2
(\BB/(\bbb\kk))^2 \}$.
Moreover, if $\kappa_{\cc}>0$ and $\bbb<1/\sqrt{2}$ we have
\[
\bigl|\supp(\widehat\beta)\bigr| \leq s \cdot \bigl( 4\cc^2 /
\kappa_{\cc} \bigr)^2 \min_{m \in\mathcal{M}^*}
\phi_{\mathrm
{max}}\bigl(m, \Gamma^{-1}\En \bigl[x_\ii
x_\ii' \bigr]\Gamma^{-1}\bigr),
\]
where $\mathcal{M}^*=\{ m \in\mathbb{N}\dvtx m > s \phi_{\mathrm
{max}}(m, \Gamma ^{-1}\En[x_\ii x_\ii']\Gamma^{-1})\cdot2 (
4\cc
^2 / \kappa_{\cc})^2 \}$.
\end{theorem}

%
%
\begin{remark}[(On the sparsity bound)] Section~\ref{cec:primitive} will
show that under minimal
and maximal sparse eigenvalues of order $s \log n$ bounded away from zero
and from above, Theorem \ref{Thm:Sparsity} implies that with a high
probability
\[
\bigl|\supp(\widehat\beta)\bigr| \lesssim s:= \bigl|\supp(\beta_0)\bigr|.\vadjust{\goodbreak}
\]
That is, the selected model's size will be of the same order as the
size of the oracle model.
We note, however, that the former condition is merely a sufficient
condition. The bound $|\supp(\widehat\beta)| \lesssim s$
will apply for other designs of interest. This can be the case even if
$\kappa_{\cc}=0$ (e.g.,
in the aforementioned design, if we change it by adding a single
repeated regressor).
\end{remark}

%
\begin{remark}[(Maximum sparse eigenvalue and sparsity)]\label{re4} Consider the
case of $f(z) = z$ with $p$ repeated regressors $x_i = (z_i,\ldots
,z_i)'$ where $|z|\leq K$. In this case, one could set $\Gamma= I\cdot
K$. In this setting, there is a sparse solution for $\LASSO$, but there
is also a solution which has all $p$ nonzero coefficients. Nonetheless,
the bound for the prediction error rate will be well behaved since $\kk
$ and $\bbb$ are invariant to the addition of copies of $z$ and
\[
\kk\geq1/K \quad\mbox{and}\quad\varrho_{c} =\bigl|\En[\varepsilon_\ii
z_\ii]\bigr|/ \bigl\{ \En \bigl[\varepsilon_\ii^2
\bigr]\En \bigl[z_\ii^2 \bigr] \bigr\}^{1/2}
\lesssim_P 1/\sqrt{n}
\]
under mild moment conditions on the noise (e.g., $\barEp[|\varepsilon
_\ii
|^3]\leq C$). In this case, $\phi_{\mathrm{max}}(m,\Gamma^{-1}\En
[x_\ii x_\ii ']\Gamma ^{-1}) = (m+1)\En[z_\ii^2]/K^2$ and the set
$\mathcal{M}$ only contains
integers larger than $p$, leading to the trivial bound $\hat m \leq p$.
\end{remark}

\subsection{Finite-sample bounds on ols post \texorpdfstring{$\LASSO$}{square-root Lasso}}\label{Sec:PostModel}
Next, we consider the ols estimator applied to the model $\widehat T$
that was selected by $\LASSO$ or includes such model (plus other
components that the data analyst may wish to include), namely $
\supp(\widehat\beta)\subseteq\widehat T$. We are interested in the case
when model selection does not work perfectly, as occurs in applications.

The following result establishes performance bounds for the ols post
$\LASSO$ estimator. Following \cite{BC-PostLASSO}, the analysis
accounts for the
data-driven choice of components and for the possibly
misspecified selected model (i.e., $ T \nsubseteq\widehat T$).

%
\begin{theorem}[(Performance of ols post $\LASSO$)]\label
{Thm:2StepNonparametric}
Under Condition \ref{condASM} and~(\ref{MainFSCondition}), let $\supp(\widehat
\beta)\subseteq\widehat T$, and $\widehat m =
|\widehat T\setminus T|$. Then we have that the ols post $\LASSO$
estimator based on $\widehat T$ satisfies
%
\[
\|\widetilde\beta- \beta_0\|_{2,n} \leq\frac{\sigma\sqrt{s+\hat m
}\|\En[x_\ii\varepsilon_\ii]\|_\infty}{\sqrt{\phi_{\mathrm
{min}}(\hat m)}}
+ 2c_s + 2\sqrt{\hat Q(\beta_0)}\BB.
\]
%
\end{theorem}

The result is derived from the sparsity of the model $\widehat T$ and
from its approximating ability. Note the presence of the new term $\|
\En
[x_\ii\varepsilon_\ii]\|_\infty$. Bounds on $\|{\En}[x_{\ii
}\varepsilon
_\ii
]\|_\infty$ can be derived using the same tools used to justify the
penalty level $\lambda$, via moderate deviation theory for
self-normalized sums \cite{jing:etal}, Gaussian approximations
to empirical processes \cite{CCK1,CCK2} or empirical process
inequalities as in \cite{BC-SparseQR}. Under mild conditions, we have
$\|\En[x_\ii\varepsilon_\ii]\|_\infty\leq C\sqrt{\log(pn) / n}$ with
probability $1-o(1)$.

\subsection{Two extreme cases}\textit{Case \textup{(i):} Parametric noiseless
case.} Consider the case that $\sigma=0$ and $c_s = 0$. Therefore, the
regression function is exactly sparse, $f(z_i)=x_i'\beta_0$. In this
case, $\LASSO$ can exactly recover the $f$ and even $\beta_0$ under
weak conditions under a broad range of penalty levels.

%
\begin{theorem}[(Exact recovery for the parametric noiseless case)]\label
{Thm:Perfect}
Under Condition \ref{condASM}, let $\sigma= 0$ and $c_s=0$. Suppose that
$\lambda
>0$ obeys the growth restriction $\bbb=\lambda\sqrt{s}/[n\kk]<1$. Then
we have $\|\hat\beta- \beta_0\|_{2,n} =0$, and if, moreover, $\kappa
_1>0$, then $\hat\beta= \beta_0$.
\end{theorem}


%
\begin{remark}[(Perfect recovery and Lasso)] It is worth mentioning that
for any $\lambda> 0$, unless
$\beta_0=0$, Lasso cannot achieve exact recovery. Moreover, it is
not obvious how to properly set the penalty level for Lasso even if we
knew a priori that it is a parametric noiseless model. In
contrast, $\LASSO$ can automatically adapt to the noiseless case.
\end{remark}

\textit{Case \textup{(ii):} Nonparametric infinite variance}. We conclude this
section with the infinite variance case. The finite sample theory does
not rely on $\Ep[\varepsilon_\ii^2]< \infty$. Instead it relies on the
choice of penalty level and penalty loadings to satisfy $\lambda/n
\geq
c\|\Gamma^{-1}\widetilde S\|_\infty$. Under symmetric errors, we
exploit the sub-Gaussian property of
self-normalized sums \cite{delapena} to develop a choice of penalty
level $\lambda$ and loadings $\Gamma= \operatorname{diag}(\gamma_j,
j=1,\ldots,p)$, where
%
%
\begin{equation}
\label{PenaltyUnbounded}\lambda= (1+u_n)c\sqrt{n} \bigl\{ 1+\sqrt {2\log(2p/
\alpha)} \bigr\} \quad\mbox{and}\quad \gamma_j = \max_{1\leq i
\leq n}
|x_{ij}|,
\end{equation}
where $u_n$ is defined below and typically we can select $u_n=o(1)$.%

%
\begin{theorem}[($\LASSO$ prediction norm for symmetric errors)]\label
{Thm:Symmetric}
Consider a nonparametric regression model with data
$(y_i,z_i)_{i=1}^n$, $y_i=f(z_i)+\sigma\varepsilon_i$, $x_i=P(z_i)$ such
that $\En[x_{\ii j}^2]=1$ ($j=1,\ldots,p$), $\varepsilon_i$'s are
independent symmetric errors, and $\beta_0$ defined as any solution to
(\ref{oracle}). Let the penalty level and loadings as in (\ref
{PenaltyUnbounded}). Assume that
there exist sequences of constants $\eta_1 \geq0$ and $\eta_2 \geq0$
both converging to $0$ and
a sequence of constants $0 \leq u_n \leq1$ such that
$P( \En[\sigma\varepsilon_\ii^2] > (1+u_n) \En[(\sigma\varepsilon
_\ii+
r_\ii)^2]) \leq\eta_1$ and $P(\En[\varepsilon_\ii^2] \leq\{
1+u_n\}^{-1}
) \leq\eta_2$ for all $n$. If $\bbb=\lambda\sqrt{s}/[n\kk]<1$, then
with probability at least $1 - \alpha- \eta_1 - \eta_2$ we have
$\lambda/n\geq c\|\Gamma^{-1}\widetilde S\|_\infty$ and
\[
\|\widehat\beta- \beta_0 \|_{2,n} \leq2\BB
\Bigl(c_s + \sigma\sqrt{\En \bigl[\varepsilon_\ii^2
\bigr]} \Bigr).
\]
\end{theorem}

The rate of convergence will be affected by how fast $\En[\varepsilon
_\ii
^2]$ diverges. That is, the final rate will depend on the particular
tail properties of the distribution of the noise. The rate also depends
on $u_n$ through $\lambda$. In many examples, $u_n$ can be chosen as a
constant or even a sequence going to zero sufficiently slowly, as in
the next corollary where $\varepsilon_i$ follows a $t$ distribution
with 2
degrees of freedom, that is, $\varepsilon_i \sim t(2)$.

%
\begin{corollary}[{[$\LASSO$ prediction norm for $\varepsilon_i\sim
t(2)$]}]\label{Thm:t(2)}
Under the setting of Theorem~\ref{Thm:Symmetric}, suppose that
$\varepsilon
_i \sim t(2)$ and are i.i.d. for all $i$. Then for any $\tau\in
(0,1/2)$, with probability at least $1-\alpha- \frac{3}{2}\tau-
\frac
{2\log(4n/\tau)}{nu_n/[1+u_n]} -\frac{72\log^2n}{n^{1/2}(\log n -
6)^2}$, we have $\lambda/n\geq c\|\Gamma^{-1}\widetilde S\|_\infty$
and, if $\bbb=\lambda\sqrt{s}/[n\kk]<1$, we have
\[
\|\widehat\beta- \beta_0 \|_{2,n} \leq2
\bigl(c_s + \sigma\sqrt{\log(4n/\tau) + 2\sqrt{2}/\tau} \bigr)\BB.
\]
\end{corollary}

%
\begin{remark}[{[Asymptotic performance in $t(2)$ case]}] Provided that
regressors are uniformly bounded and satisfy the sparse eigenvalues
condition (\ref{SparseEigAssump}), we have that the restricted
eigenvalue $\kappa_{\cc}$ is bounded away from zero for the specified
choice of $\Gamma$. Because Corollary~\ref{Thm:t(2)} ensures $\lambda/n
\geq c\|\Gamma^{-1}\widetilde S\|_\infty$ with the stated probability,
by Lemmas \ref{Lemma:NewKappa} and \ref{Lemma:NewRho} we have
\[
\varrho_{c}+\bbb\lesssim\frac{\lambda\sqrt{s}}{n\kappa_{\cc}} \lesssim(1+u_n)
\sqrt{\frac{s\log(p\vee n) }{n}}\quad\implies\quad\BB\lesssim\sqrt{\frac
{s\log(p\vee n) }{n}}.
\]
Therefore, under these design conditions, assuming that $s\log
(p/\alpha
) = o(n)$ and that $\sigma$ is fixed, and setting $1/\alpha= o(\log
n)$, we can select $u_n = 1/2$ and $\tau= 1/\log n$ in Corollary~\ref
{Thm:t(2)}, to conclude that the $\LASSO$ estimator satisfies
%
%
\begin{equation}
\label{t2Bound} \| \widehat\beta- \beta_0 \|_{2,n}
\lesssim(c_s + \sigma\sqrt{\log n})\sqrt{\frac{s\log(p\vee n)}{n}},
\end{equation}
with probability $1-\alpha(1+o(1))$. Despite the infinite variance, the
bound (\ref{t2Bound}) differs from the Gaussian noise case only by a
$\sqrt{\log n}$ factor.
\end{remark}

\section{Asymptotics analysis under primitive conditions}\label{cec:primitive}

In this section, we formally state an algorithm to compute the
estimators and we provide rates of convergence results under simple
primitive conditions. 

We propose setting the penalty level as
%
%
\begin{equation}
\label{Eq:Def-lambda}\lambda= c \sqrt{n} \Phi^{-1}(1-\alpha/2p),
\end{equation}
where $\alpha$ controls the confidence level, and $c>1$ is a slack
constant similar to \cite{BickelRitovTsybakov2009}, and the penalty
loadings according to the following iterative algorithm.

%
\begin{algorithm}[(Estimation of square-root Lasso loadings)]\label{alg1} Choose
$\alpha\in(1/n,1/2]$ and a constant $K\geq1$ as an upper bound on
the number of iterations. (0) Set $k=0$, $\lambda$ as in (\ref
{Eq:Def-lambda}), and $\hat\gamma_{j,0} = \max_{1\leq i\leq n}
|x_{ij}|$ for each $j=1,\ldots,p$. (1) Compute the $\LASSO$ estimator
$\hat\beta$ based on the current penalty loadings $\Gamma= \hat
\Gamma
_k = \operatorname{diag}\{\hat\gamma_{j,k}, j=1,\ldots,p\}$. (2) Set
\[
\hat\gamma_{j,k+1}:= 1 \vee\sqrt{\En \bigl[
x_{\ii j}^2 \bigl(y_\ii-x_\ii'
\hat\beta \bigr)^2 \bigr]}/\sqrt{\En \bigl[
\bigl(y_\ii-x_\ii' \hat\beta
\bigr)^2 \bigr]}.
\]
(3) If $k> K$, stop; otherwise set $k \leftarrow k+1 $ and go to step 1.
\end{algorithm}

%
\begin{remark}[(Parameters of the algorithm)]The parameter $1-\alpha$ is
a confidence level which guarantees
near-oracle performance with probability at least $1-\alpha$; we
recommend $\alpha= 0.05/\log n$. The constant $c>1$ is the slack
parameter used as in \cite{BickelRitovTsybakov2009}; we
recommend $c=1.01$. In order to invoke moderate deviation theorem for
self-normalized sums, we need to be able to bound with a high
probability:
%
%
\begin{equation}
\label{Motivation} \sqrt{\En \bigl[x_{\ii j}^2
\varepsilon_\ii^2 \bigr]}/\sqrt{\En \bigl[
\varepsilon_\ii^2 \bigr]} \leq\gamma_{j,0}.
\end{equation}
The choice of $\hat\gamma_{j,0}= \max_{1\leq i\leq n}|x_{ij}|$
automatically achieves (\ref{Motivation}). Nonetheless, we recommend
iterating the procedure to avoid unnecessary over-penalization, since
at each iteration more precise estimates of the penalty loadings are
achieved. These recommendations are valid either in finite or large
samples under the conditions stated below. They are also supported by
the numerical experiments (see Section~G of supplementary material \cite{BCWsupp}).
\end{remark}

%
\begin{remark}[(Alternative estimation of loadings)]\label
{Remark:AlternativeLoadings}
Algorithm \ref{alg1} relies on the $\LASSO$ estimator $\hat\beta$. Another
possibility is to use the ols post $\LASSO$ estimator~$\widetilde
\beta
$. This leads to similar theoretical and practical results. Moreover,
we can define the initial penalty loading as $\hat\gamma_{j,0} = W \{
\En[x_{\ii j}^4]\}^{1/4}$ where the kurtosis parameter $W>\{\barEp
[\varepsilon_\ii^4]\}^{1/4}/\{\barEp[\varepsilon_\ii^2]\}^{1/2}$
is pivotal
with respect to the scaling parameter $\sigma$, but we need to assume
an upper bound for this quantity. The purpose of this parameter is to
bound the kurtosis of the marginal distribution of errors, namely that
of $\bar F_\varepsilon(v) = n^{-1} \sum_{i=1}^n P( \varepsilon_i
\leq v)$. We
recommend $W=2$, which permits a wide class of marginal distributions
of errors, in particular it allows $\bar F_\varepsilon$ to have tails as
heavy as those of $t(a)$ with $a>5$. This method also achieves (\ref
{Motivation}); see Section C.1 of supplementary material \cite{BCWsupp}.
\end{remark}

The following is a set of simple sufficient conditions which yields
practical corollaries.
Let $\ell_n \nearrow\infty$ be a sequence of positive constants.

\renewcommand{\thecond}{P}
\begin{cond}\label{condP}
The noise and regressors obey $\sup_{n\geq1} \barEp[|\varepsilon_\ii|^q] <
\infty$, $q>4$, $\inf_{n\geq1} \min_{1\leq j\leq p} \En[x_{\ii
j}^2\Ep
[\varepsilon_\ii^2]] > 0$, $\sup_{n\geq1} \max_{1\leq j\leq p}
\mathbb{E}_n[|x_j|^{3}\mathrm{E}[|\varepsilon|^{3}]] <\infty$ and
%
%
\begin{equation}
\label{SparseEigAssump}\sup_{n\geq1} \phi_{\mathrm{max}}\bigl(s\ell
_n,\En \bigl[x_\ii x_\ii'
\bigr]\bigr)/\phi_{\mathrm{min}}\bigl(s \ell _n,\En
\bigl[x_\ii x_\ii' \bigr]\bigr) <\infty.
\end{equation}
Moreover, we have that\vspace*{2pt} $\max_{i \leq n, j \leq p} |x_{ij}|^2/\ell_n
=o(1)$, $\log p \leq C(n/\log^2 n)^{1/3}$, $\ell_n^2s\log(p\vee n)
\leq C n/\log n$, $s\geq1$, and $c^2_s \leq C\sigma^2( s \log(p \vee
n) /n)$.
\end{cond}


Condition \ref{condP} imposes conditions on moments that allow us to use
results of the moderate deviation theory for self-normalized sums, weak
requirements on $(s,p,n)$, well behaved sparse eigenvalues as a
sufficient condition on the design to bound the impact factors
and a mild condition on the approximation errors (see Remark~\ref
{Remark:ApproximationError} for a discussion and references).

The proofs in this section rely on the following result due to \cite
{jing:etal}.

%
\begin{lemma}[(Moderate deviations for self-normalized
sums)]\label{Lemma: MDSN} Let $X_{1},\ldots, X_{n}$ be independent,
zero-mean random variables and $\delta\in(0,1]$. Let
$S_{n,n} = n\En[X_\ii]$, $V^2_{n,n} = n\En[X^2_\ii], M_n = \{\barEp
[|X_\ii
|^{2+\delta}]\}^{1/\{2+\delta\}}/\{\barEp[ X_\ii^2 ] \}^{1/2}<\infty
$ and $\jmath_n \leq\break   n^{{\delta}/{(2(2+\delta))}}M^{-1}_n$. For some
absolute constant $A$, uniformly on
$
0 \leq|x| \leq\break   n^{{\delta}/{(2(2+\delta))}}M^{-1}_n/\jmath_n-1$,
we have
\[
\biggl\vert\frac{P(S_{n,n}/V_{n,n} \geq x) }{ (1-\Phi(x))} - 1 \biggr\vert\leq\frac{A}{ \jmath_n^{2+\delta}}.
\]
\end{lemma}

The following theorem summarizes the asymptotic performance of $\LASSO
$, based upon Algorithm \ref{alg1}, for commonly used designs.

%
\begin{theorem}[(Performance of $\LASSO$ and ols post $\LASSO$ under
Condition~\ref{condP})]\label{CorFinite:SqrtLASSO}
Suppose Conditions \ref{condASM} and \ref{condP} hold. Let $\alpha\in(1/n,1/\log n)$,
$c\geq1.01$, the penalty level $\lambda$ be set as in (\ref
{Eq:Def-lambda}) and the penalty loadings as in Algorithm \ref{alg1}. Then for
all $n\geq n_0$, with probability at least $ 1 - \alpha\{ 1 + \bar C /
\log n\} - \bar C \{n^{-1/2}\log n + n^{1-q/4}\}$ we have
\begin{eqnarray*}\|\widehat\beta- \beta_0
\|_{2,n} &\leq&\sigma\bar C \sqrt{ \frac{ s \log(n\vee(p/\alpha
))}{n} }, \\
 \sqrt{
\En \bigl[ \bigl(f_\ii-x_\ii'\hat\beta
\bigr)^2 \bigr]} &\leq&\sigma\bar C \sqrt{ \frac{ s \log
(n\vee
(p/\alpha))}{n} },
\\
\|\widehat\beta-\beta_0\|_1 &\leq&\sigma\bar C \sqrt{
\frac{ s^2
\log
(n\vee(p/\alpha))}{n} } \quad \mbox{and}\quad\bigl |\supp(\widehat\beta)\bigr| \leq \bar C s,
\end{eqnarray*}
where $n_0$ and $\bar C$ depend only on the constants in Condition \ref{condP}.
Moreover, the ols post $\LASSO$ estimator satisfies with the same
probability for all $n\geq n_0$,
\begin{eqnarray*}\|\widetilde\beta- \beta_0
\|_{2,n} &\leq&\sigma\bar C \sqrt{ \frac{ s \log(n\vee(p/\alpha
))}{n} }, \\
\sqrt{\En
\bigl[ \bigl(f_\ii-x_\ii'\widetilde\beta
\bigr)^2 \bigr]} &\leq&\sigma\bar C\sqrt{ \frac{ s \log
(n\vee(p/\alpha))}{n} } \quad\mbox{and} \\
\|\widehat\beta-\beta_0\|_1 &\leq&\sigma\bar C
\sqrt{ \frac{ s^2
\log
(n\vee(p/\alpha))}{n} }.
\end{eqnarray*}
\end{theorem}

%
\begin{remark}[(Gaussian-like performance and normalization assumptions)]
Theorem~\ref{CorFinite:SqrtLASSO} yields bounds on the estimation
errors that are ``Gaussian-like,'' namely the factor $\sqrt{\log
(p/\alpha)}$ and other constants in the performance bound are the same
as if errors were Gaussian, but the probabilistic guarantee is not
$1-\alpha$ but rather $1-\alpha+o(1)$, which together with mildly more
restrictive growth conditions is the cost of non-Gaussianity. We also
note that the normalization $\En[x_{j}^2]=1$, $j=1,\ldots,p$ is not used
in the construction of the estimator, and the results of the theorem
hold under the condition: $C_1 \leq\En[x_{j}^2] \leq C_2, j=1,\ldots,p$
uniformly for all $n \geq n_0$, for some positive, finite constants
$C_1$ and $C_2$.
\end{remark}

The results above establish that $\LASSO$ achieves the same near oracle
rate of convergence as Lasso despite not knowing the scaling parameter
$\sigma$. They allow for heteroscedastic errors with mild restrictions
on its moments. Moreover, it allows for an arbitrary number of
iterations. The results also establish that the upper bounds on
the rates of convergence of $\LASSO$ and ols post $\LASSO$ coincide
under these conditions. This is confirmed also by Monte--Carlo experiments
reported in the supplementary material \cite{BCWsupp}, with ols post $\LASSO$ performing no
worse and often outperforming $\LASSO$
due to having a much smaller bias. Notably, this theoretical and
practical performance occurs despite the fact that $\LASSO$ may in
general fail to correctly select the oracle model $T$ as a subset and
potentially select variables not in~$T$.

\renewcommand{\theexample}{S}
\begin{example}[(Performance for Sobolev balls and $p$-rearranged
Sobolev balls)]\label{exS}
In this example, we show how our results apply to an
important class of Sobolev functions, and illustrates how modern
selection drastically reduces the dependency on knowing the order of
importance of the basis functions.

Suppose that $z_i$'s are generated as i.i.d. from $\operatorname{Uniform}(0,1)$, $x_i$'s are
formed as $(x_{ij})_{j=1}^p$ with $x_{ij} = P_j(z_i)$, $\sigma= 1$,
and $\varepsilon_i \sim N(0,1)$. Following \cite{TsybakovBook}, consider
an orthonormal bounded basis $\{P_j(\cdot)\}_{j=1}^\infty$ in
$L^2[0,1]$, consider functions $f(z)=\sum_{j=1}^\infty\theta_jP_j(z) $
in a Sobolev space $\mathcal{S}(\alpha,L)$ for some $\alpha\geq1$
and $L>0$.
This space consists of functions whose Fourier coefficients $\theta$
satisfy $\sum_{j=1}^\infty|\theta_j| < \infty$ and
\[
\theta\in\Theta(\alpha, L) = \Biggl\{ \theta\in\ell^2(
\mathbb{N})\dvtx%
\sum
_{j=1}^\infty j^{2\alpha} \theta_j^2
\leq L^2
\Biggr\}.
\]

We also consider functions in a $p$-rearranged Sobolev space $\mathcal
{RS}(\alpha,p,L)$. These functions take the form $f(z)=\sum_{j=1}^\infty
\theta_jP_j(z) $\vadjust{\goodbreak} such that $\sum_{j=1}^\infty|\theta_j| < \infty$ and
$\theta\in\Theta^{R}(\alpha,p, L)$, where
\[
\Theta^{R}(\alpha,p,L) = \lleft\{ \theta\in\ell^2(
\mathbb{N}) \dvtx\quad
\begin{array} {l}\displaystyle\exists
\mbox{ permutation} \Upsilon \dvtx\{1,\ldots,p\} \to\{1,\ldots,p\}
\\[3pt]
\displaystyle\sum_{j=1}^p j^{2\alpha}
\theta_{\Upsilon
(j)}^2 + \sum_{j=p+1}^\infty
j^{2\alpha}\theta_{j}^2 \leq L^2
\end{array} %
\rright\}.
\]
Note that $\mathcal{S}(\alpha,L) \subset\mathcal{RS}(\alpha,p,L)$.

In the supplementary material \cite{BCWsupp}, we show that the rate-optimal choice for the
size of the support of the oracle model $\beta_0$ is $s\lesssim
n^{1/[2\alpha+1]}$. This rate can be achieved with the support
consisting of indices $j$ that correspond to the $s$ largest
coefficients $|\theta_j|$. The oracle projection estimator $\hat\beta
^{\mathrm{or}}$ that uses these ``ideal'' $s$ components achieves
optimal prediction error rate uniformly over the regression functions
$f \in\mathcal{S}(\alpha,L)$ or $f \in\mathcal{RS}(\alpha,p,L)$:
$ (\En[\{f_\ii-\sum_{j=1}^\infty\hat\beta^\mathrm{or}_jP_j(z_\ii
)\}
^2])^{1/2} \lesssim_P n^{-\alpha/[2\alpha+1]}$.
Under mild regularity conditions, as in Theorem~\ref
{CorFinite:SqrtLASSO}, $\LASSO$ estimator $\hat\beta$
that uses $x_i = (P_1(z_i),\ldots, P_p(z))'$ achieves a near-optimal rate
uniformly over the regression functions $f \in\mathcal{S}(\alpha,L)$
or $f \in\mathcal{RS}(\alpha,p,L)$:
\[
\sqrt{\En \bigl[ \bigl(f_\ii-x_\ii'
\hat\beta \bigr)^2 \bigr]} \lesssim_P n^{-\alpha
/[2\alpha
+1]}
\sqrt{\log(n\vee p)},
\]
without knowing the ``ideal'' $s$ components among $x_i$. The same
statement also holds for the ols post $\LASSO$
estimator $\widetilde\beta$.

Therefore, the $\LASSO$ and ols post $\LASSO$ estimators achieve near
oracle rates uniformly over
rearranged Sobolev balls under mild conditions. As a contrast,
consider the ``naive oracle'' series projection estimator that uses the
first $s$ components of the basis, assuming that the parameter space is
$\mathcal{S}(\alpha,L)$. This estimator achieves the optimal rate for
the Sobolev space $\mathcal{S}(\alpha,L)$, but fails to be uniformly
consistent over $p$-rearranged Sobolev space $\mathcal{RS}(\alpha
,p,L)$, since we can select a model $f \in\mathcal{RS}(\alpha,p,L)$
such that its first $s$ Fourier coefficients are zero,
and the remaining coefficients are nonzero, therefore, the ``naive
oracle'' fit will be $0$ plus a centered noise, and the estimator will
be inconsistent for this $f$.
\end{example}

We proceed to state a result on estimation of $\sigma^2$ under the
asymptotic framework.

%
%
\begin{corollary}[(Estimation of $\sigma^2$ under asymptotics)]\label
{Thm:SigmaInference}
Suppose Conditions~\ref{condASM} and \ref{condP} hold. Let $\alpha\in(1/n,1/\log n)$,
$c\geq1.01$, the penalty level $\lambda$ be set as in (\ref
{Eq:Def-lambda}) and the penalty loadings as in Algorithm \ref{alg1}. Then for
all $n\geq n_0$, with probability at least $ 1 - \alpha\{ 1 + \bar C /
\log n\} - \bar C \{n^{-1/2}\log n +n^{1-q/4}\} - 2\delta$,
\[
\bigl\vert\widehat Q(\widehat\beta)- \sigma^2 \bigr\vert\leq
\frac
{\sigma^2
\bar C s \log(n\vee(p/\alpha))}{n} + \frac{\sigma^2 \bar C \sqrt {s\log
(p\vee n)}}{\sqrt{\delta}n^{1-1/q}} + \frac{\sigma^2 \bar C}{\sqrt {\delta n}}.
\]
Moreover, provided further that $s^2\log^2(p\vee n) \leq Cn/\log n$, we
have that
\[
 \bigl\{\sigma^2\xi_n\bigr\}^{-1} n^{1/2} \bigl( \widehat Q(\widehat\beta)-
\sigma
^2 \bigr) \Rightarrow N(0,1),
\]
where $\xi_n^2 = \barEp[\{\varepsilon
_\ii
^2-\Ep[\varepsilon_\ii^2]\}^2]$.
\end{corollary}

This result extends \cite{BC-PostLASSO,SunZhang2012} to the
heteroscedastic, non-Gaussian cases.

\section{An application to a generic semi-parametric inference problem}\label{sec5}

In this section, we present a generic application of the methods of
this paper to semiparametric
problems, where some lower-dimensional structural parameter is of
interest and the $\LASSO$ or ols post $\LASSO$
are used to estimate the high-dimensional nuisance function. We denote
the true
value of the target parameter by $\theta_0 \in\Theta\subset\RR
^{d}$, and assume that
it satisfies the following moment condition:
%
%
\begin{equation}
\label{eq:ivequation} \Ep \bigl[ \psi \bigl(w_i, \theta_0,
h_0(z_i) \bigr) \bigr] = 0,\qquad i =1,\ldots,n,
\end{equation}
where $w_i$ is a random vector taking values in $\mathcal{W}$,
containing vector $z_i$ taking values
in $\mathcal{Z}$ as a subcomponent; the function $(w, \theta, t)
\mapsto\psi(w, \theta, t) = (\psi_j(w, \theta, t))_{j=1}^d$ is a
measurable map from an open neighborhood of $ \mathcal{W} \times
\Theta
\times T $, a subset of $\RR^{d_w + d+ d_t}$, to $\RR^{d}$, and $z
\mapsto h_0(z) = (h_{0m}(z))_{m=1}^M$ is the nuisance function mapping
$\mathcal{Z}$ to
$T \subset\RR^M$. We note that $M$ and $d$ are fixed and do not
depend on $n$ in what follows.

Perhaps the simplest, that is linear,
example of this kind arises in the instrumental variable (IV)
regression problem in \cite
{BellChenChernHans:nonGauss,BellChernHans:Gauss}, where $ \psi(w_i,
\theta_0, h_0(z_i) ) = (u_i - \theta_0 d_i) h_0(z_i)$, where $u_i$ is
the response variable, $d_i$ is the endogenous variable,
$z_i$ is the instrumental variable, $h_0(z_i) = \Ep[d_i \mid z_i]$ is
the optimal instrument, and $\Ep[(u_i - \theta_0 d_i) \mid z_i] =0$.
Other examples include partially linear models, heterogeneous treatment
effect models, nonlinear instrumental variable, $Z$-problems as well as
many others
(see, e.g., \cite
{amemiya1977maximum,huber1967behavior,hansen1982large,chamberlain,BellChenChernHans:nonGauss,BCH2011:InferenceGauss,c.h.zhang:s.zhang,BelloniChernozhukovHansen2011,gautier:tsybakov,vandeGeerBuhlmannRitov2013,BCK-LAD,ren2013asymptotic,BelloniChernozhukovWei2013,Farrell:JMP,BCFH-LATE}), which all give rise to nonlinear moment conditions with
respect to the nuisance functions.

We assume that the nuisance functions $h_{0}$ arise as conditional
expectations of some variables that can be modeled and estimated in the
approximately sparse framework, as formally described below. For
instance, in the example mentioned above, the function $h_0$ is indeed
a conditional expectation of the endogenous variable given the
instrumental variable. We let
$\hat h= (\hat h_{m})_{m=1}^M$ denote\vspace*{1pt} the estimator of $h_0$, which
obeys conditions stated below. The estimator $\hat\theta$ of $\theta
_0$ is constructed as any approximate $\varepsilon_n$-solution in
$\Theta$
to a sample analog of the moment condition above:
%
%
\begin{equation}
\label{eq:analog} \bigl\|\En \bigl[ \psi \bigl(w, \hat\theta,
\hat h(z) \bigr) \bigr] \bigr\|
\leq\varepsilon_n\qquad \mbox{where } \varepsilon_n = o
\bigl(n^{-1/2} \bigr).
\end{equation}

The key condition needed for regular estimation of $\theta_0$ is the
orthogonality condition:
%
%
\begin{equation}
\label{eq:orthogonality} \Ep \bigl[ \partial_t \psi \bigl(w_i,
\theta_0, h_0(z_i) \bigr) |z_i
\bigr] = 0, \qquad i =1,\ldots,n,
\end{equation}
where here and below we use the symbol $\partial_t$ to abbreviate
$\frac
{\partial}{\partial t'}$.
For instance, in the IV example this condition holds, since $\partial_t
\psi(w_i, \theta_0, h_0(z_i) ) = (u_i - \theta_0 d_i)$ and $\Ep
[(u_i -
\theta_0 d_i)|z_i] = 0$ by assumption. In other examples, it is
important to construct the scores
that have this orthogonality property. Generally, if we have a score,
which identifies the target parameter but does not have the
orthogonality property, we can construct the score that has the
required property by projecting the original score onto the
orthocomplement of the tangent space for the nuisance parameter; see,
for example, \cite{vdV-W,vdV,kosorok:book} for detailed discussions. This
often results in a semiparametrically efficient score function.

The orthogonality condition reduces sensitivity to ``crude'' estimation
of the nuisance function $h_0$.
Indeed, under appropriate sparsity assumptions stated below, the
estimation errors for $h_0$,
arising as sampling, approximation, and model selection errors, will be
of order $o_P( n^{-1/4})$.
The orthogonality condition together with other conditions will
guarantee that these estimation
errors do not impact the first-order asymptotic behavior of the
estimating equations, so that
%
%
\begin{equation}
\sqrt{n} \En \bigl[ \psi \bigl(w, \hat\theta, \hat h(z) \bigr) \bigr] = \sqrt{n}
\En \bigl[ \psi \bigl(w, \hat\theta, h_0(z) \bigr) \bigr] +
o_P(1).
\end{equation}
This leads us to a regular estimation problem, despite $\hat h$ being
highly nonregular.

In what follows, we shall denote by $c$ and $C$ some positive
constants, and by $L_n$ a sequence of positive
constants that may grow to infinity as $n \to\infty$.

\renewcommand{\thecond}{SP}
\begin{cond}\label{condSP}
For each $n$, we observe the independent
data vectors $(w_i)_{i=1}^n$ with law determined by the probability
measure $\mathrm{P}=\mathrm{P}_n$. Uniformly, for all n the following
conditions hold. (i) The true parameter values $\theta_0$ obeys (\ref
{eq:ivequation}) and is interior relative to $\Theta$, namely there is
a ball of fixed positive radius centered at $\theta_0$ contained in
$\Theta$, where $\Theta$ is a fixed compact subset of $\RR^d$. (ii) The
map $\nu\mapsto\psi(w, \nu)$ is twice continuously differentiable
with respect to $\nu= (\nu_k)_{k=1}^K = (\theta, t )$ for all $\nu
\in
\Theta\times T$, where $T$ is convex, with
$\sup_{\nu\in\Theta\times T} | \partial_{\nu_k}\partial_{\nu_r}
\psi
_j(w_i, \nu)| \leq L_n$ a.s., for all $k, r \leq K$, $j\leq d$, and
$i\leq n$. The conditional second moments of the first derivatives are
bounded as follows: $\mathrm{P}$-a.s. $ \Ep( \sup_{\nu\in\Theta
\times T}| \partial_{\nu_k} \psi_j(w_i, \nu)|^2 \mid z_i) \leq C$ for
each $k$, $j$ and $i$. (iii)~The orthogonality condition (\ref
{eq:orthogonality}) holds. (iv) The following identifiability condition
holds: for all $\theta\in\Theta$, $\|\barEp[\psi(w, \theta, h_0(z))]
\| \geq2^{-1} ( \|J_n (\theta- \theta_0)\| \wedge c)$,
where $J_n:= \barEp[\partial_\theta\psi(w, \theta_0, h_0(z))] $ has
singular values bounded away from zero and above. (v) $\barEp[\|\psi
(w, \theta_0, h_0(z))\|^{3}]$ is bounded from above.
\end{cond}

In addition to the previous conditions, Condition \ref{condSP} imposes standard
identifiability and certain smoothness on the problem, requiring second
deriv\-atives to be bounded by $L_n$, which is allowed to grow with $n$
subject to restrictions specified below. It is possible to allow for
nondifferentiable $\psi$ at the cost of a more complicated argument;
see \cite{BCK-LAD}. In what follows, let $\delta_n \searrow0$ be a
sequence of constants approaching zero from above.

\renewcommand{\thecond}{AS}
\begin{cond}\label{condAS}
The following conditions hold for each
$n$. (i) The function $h_0=(h_{0m})_{m=1}^M\dvtx\mathcal{Z}\mapsto T$ is
approximately sparse, namely, for each $m$, $h_{0m}(\cdot) $= $\sum_{l=1}^p P_{ml}(\cdot) \beta_{0ml} + r_{m}(\cdot)$,
where $P_{ml}\dvtx\mathcal{Z} \mapsto\RR$ are approximating functions,
$ \beta_{0m} = (\beta_{0ml})_{l=1}^p$ obeys $|\supp(\beta_{0m})|
\leq
s$, $s \geq1$, and the approximation errors $(r_m)_{m=1}^M\dvtx
\mathcal
{Z} \to\RR$ obey $\barEp[r^2_m(z)] \leq C s \log(p \vee n)/n$. There
is an estimator $\hat h_{m}(\cdot)= \sum_{l=1}^p P_{ml}(\cdot) \hat
\beta_{ml}$ of $h_{0m}$ such that, with probability at least $1-\delta
_n$, $\hat h = (\hat h_{m})_{m=1}^M$ maps $\mathcal{Z}$ into $T$,
$\hat
\beta_{m} = (\hat\beta_{ml})_{l=1}^p$ satisfies $\|\hat\beta_m -
\beta
_{0m}\|_1 \leq C \sqrt{ s^2 \log(p \vee n)/n}$ and $\En[(\hat h_{m}(z)
- h_{0m}(z))^{2}]\leq C s \log(p \vee n)/n$ for all $m$. (ii) The
scalar random variables $\dot\psi_{mjl}(w_i):= \partial_{t_m} \psi_j
(w_i, \theta_0, h_0(z)) P_{ml}(z_i)$ obey $\max_{m,j,l} \En[|\dot
\psi
_{mjl}(w) |^2] \leq L^2_n$
with probability at least $1-\delta_n$ and $\max_{m,j,l}$ $(\barEp
[|\dot\psi_{mjl}(w) |^3])^{1/3}/(\barEp[|\dot\psi_{mjl}(w)
|^2])^{1/2} \leq M_n$.
(iii) Finally,
the following growth restrictions hold as $n \to\infty$:
%
%
\begin{equation}
\label{key growth} L^2_n s^2 \log^2
(p \vee n)/ n \to0 \quad\mbox{and}\quad \log(p\vee n) n^{-1/3}
M_n^2 \to0.
\end{equation}
\end{cond}

The assumption records a formal sense in which approximate sparsity is
used, as well as requires reasonable behavior of the estimator $\hat
h$. In the previous
sections, we established primitive conditions under which this behavior
occurs in problems
where $h_0$ arise as conditional expectation functions. By virtue of
(\ref{key growth}) the assumption implies that $\{\En(\hat h_{m}(z) -
h_{0m}(z))^2\}^{1/2} = o_P(n^{-1/4})$. It is standard that the
square of this term multiplied by $\sqrt{n}$ shows up as a
linearization error for $\sqrt{n}(\hat\theta- \theta_0)$ and,
therefore, this term does not affect its first-order behavior.
Moreover, the assumption implies by virtue of (\ref{key growth}) that
$\|\hat\beta_m - \beta_{0m}\|_1 = o_P(L_n^{-1}( \log(p \vee
n))^{-1})$, which is used to control
another key term in the linearization as follows:
\[
\sqrt{n} \max_{j,m,l} \bigl|\En \bigl[ \dot\psi_{mjl} (w)
\bigr]\bigr| \| \hat\beta_m - \beta_{0m}\|_1
\lesssim_P L_n \sqrt{\log(p \vee n)} \| \hat\beta
_m - \beta_{0m}\|_1 = o_P(1),
\]
where the bound follows from an application of the moderate deviation
inequalities for self-normalized sums (Lemma~\ref{Lemma: MDSN}). The
idea for this type of control is borrowed from \cite
{BellChenChernHans:nonGauss}, who used it in the IV model above.

%
\begin{theorem}\label{theorem:semiparametric} Under Conditions \ref{condSP} and
\ref{condAS}, the estimator $\hat\theta$ that obeys equation (\ref{eq:analog})
and $\widehat\theta\in\Theta$ with probability approaching 1, satisfies
$ \sqrt{n}(\hat\theta- \theta_0) = - J^{-1}_n \frac{1}{\sqrt {n}}\sum_{i=1}^n \psi(w_i, \theta_0, h_0(z_i)) + o_P(1)$.
Furthermore, provided
$\Omega_n = \barEp[\psi(w, \theta_0, h_0(z)) \psi(w, \theta_0,
h_0(z))']$ has eigenvalues bounded
away from zero,
\[
\Omega_n^{-1/2}J_n \sqrt{n}(\hat\theta-
\theta_0) \Rightarrow N(0,I).\vadjust{\goodbreak}
\]
\end{theorem}
This theorem extends the analogous result in \cite
{BellChernHans:Gauss,BellChenChernHans:nonGauss} for a specific linear
problem to a generic nonlinear setting, and could be of independent
interest in many problems cited above.
The theorem allows for the probability measure $\mathrm{P}=\mathrm
{P}_n$ to change with $n$, which implies that the confidence bands
based on the result have certain uniform validity (``honesty'') with
respect to $\mathrm{P}$, as formalized in \cite
{BelloniChernozhukovHansen2011}, thereby constructively addressing
\cite
{leeb:potscher:pms}'s critique.
See \cite{BCK-LAD} for a generalization of the result above to the case
$\operatorname{dim}(\theta_0)\gg n$.

\renewcommand{\theexample}{PL}
\begin{example}\label{exPL}
It is instructive to conclude this section by
inspecting the example of approximately sparse partially linear
regression \cite{BelloniChernozhukovHansen2011,BCH2011:InferenceGauss},
which also nests the sparse linear regression model \cite
{c.h.zhang:s.zhang}. The partially linear model of \cite{robinson}
is
\begin{eqnarray*}
y_i &=& d_i \theta_0 + g(z_i)
+ \varepsilon_i,\qquad \mathrm{E}[\varepsilon_i|z_i,
d_i] = 0,\\
 d_i &=& m(z_i) + v_i,\qquad
\mathrm{E}[v_i|z_i] = 0.
\end{eqnarray*}
The target is the real parameter $\theta_0$, and an orthogonal score
function $\psi$ for this parameter is
$
\psi(w_i, \theta, t) = (y_i - \theta(d_i - t_2 ) - t_1) (d_i - t_2)$,
where $t = (t_1, t_2)'$ and $w_i= (y_i, d_i, z_i')'$. Let $\ell(z_i):=
\theta_0 m(z_i) + g(z_i)$, and $h_0(z_i):= ( \ell(z_i),\break m(z_i))' =
(\mathrm{E}[y_i|z_i],  \mathrm{E}[d_i|z_i])' $. Note that
\[
\mathrm{E} \bigl[\psi \bigl(w_i, \theta_0,
h_0(z_i) \bigr)|z_i \bigr] = 0 \quad\mbox{and}\quad
 \mathrm{E} \bigl[\partial_t \psi \bigl(w_i, \theta,
h_0(z_i) \bigr)|z_i \bigr] = 0,
\]
so
the orthogonality condition holds. If the regression functions $\ell(z_i)$
and $m(z_i)$ are approximately sparse with respect to
$x_i = P(z_i)$, we can estimate them by $\LASSO$
or ols post $\LASSO$ regression of $y_i$ on $x_i$
and $d_i$ on $x_i$, respectively. The resulting
estimator $\hat\theta$ of $\theta_0$, defined as a solution to (\ref
{eq:analog}), is a $\LASSO$ analog of Robinson's \cite{robinson}
estimator. If assumptions of Theorem~\ref{theorem:semiparametric} hold,
$\hat\theta$ obeys
\[
\Omega_n^{-1/2} J_n \sqrt{n}(\hat\theta-
\theta
_0) \Rightarrow N(0,1)
\]
for $J_n = \bar{\mathrm{E}}[v^2]$ and $\Omega_n
= \bar{\mathrm{E}}[\varepsilon^2 v^2]$. In a
homoscedastic model, $\hat\theta$ is semiparametrically efficient,
since its asymptotic variance $\Omega_n/J_n^2$ reduces to the
efficiency bound $\mathrm{E}[\varepsilon^2]/\mathrm{E}[v^2]$ of Robinson
\cite{robinson}; as pointed out in \cite
{BelloniChernozhukovHansen2011,BCH2011:InferenceGauss}. In the linear
regression model, this estimator is first-order equivalent to, but
different in finite samples from, a one-step correction from the scaled
Lasso proposed in \cite{c.h.zhang:s.zhang}; in the partially linear
model, it is equivalent to the post-double selection method of~\mbox{\cite
{BelloniChernozhukovHansen2011,BCH2011:InferenceGauss}}.
\end{example}

\begin{appendix}\label{app}
\section{Proofs of Section~\texorpdfstring{\lowercase{\protect\ref{sec:finite-result}}}{3}}
\begin{pf*}{Proof of Lemma~\ref{Lemma:NewKappa}}
The first result holds by the inequalities given in the main text.


To show the next statement, note that $T$ does not change by including
repeated regressors (indeed, since $T$ is selected by the oracle (\ref
{oracle}), $T$ will not contain repeated regressors). Let $R$ denote
the set of repeated regressors and $\tilde x_i = (x_i',z_i')'$ where
$x_i\in\RR^p$ is the vector of original regressors and $z_i\in\RR
^{|R|}$ the vector of repeated regressors. We denote by $\widetilde
\Gamma$ and $\|\cdot\|_{2,\tilde n}$ the penalty loadings and the
prediction norm associated with $(\tilde x_i)_{i=1}^n$. Let $\tilde
\delta= ({\delta^1}',{\delta^2}')'$, where $\delta^1 \in\RR^p$ and
$\delta^2\in\RR^{|R|}$, define $\bar\delta^2 \in\RR^p$ so that
$\bar
\delta^2_j = \delta^2_j $ if $j\in R$, and $\bar\delta^2_j=0$ if
$j\notin R$, and denote $\delta= \delta^1 + \bar\delta^2$. It
follows that
\[
\lkk\geq\frac{\|\tilde\delta\|_{2,\tilde n}}{\|\widetilde\Gamma
\tilde
\delta_T\|_1-\|\widetilde\Gamma\tilde\delta_{T^c}\|_1}=\frac{\|
\delta\|
_{2,n}}{\|\Gamma\delta_T^1\|_1-\|\Gamma\delta_{T^c}^1\|_1-\|\Gamma
\bar
\delta_{T}^2\|_1-\|\Gamma\bar\delta_{T^c}^2\|_1},
\]
which is minimized in the case that $\delta^1 = \delta$ and $\bar
\delta
^2 =0$. Thus, the worst case for $\kk$ correspond to $\bar\delta^2 =0$
which corresponds to ignoring the repeated regressors.
\end{pf*}

\begin{pf*}{Proof of Lemma~\ref{Lemma:NewRho}} The first part is shown in
the main text. The second part is proven in supplementary material \cite{BCWsupp}.
\end{pf*}

\begin{pf*}{Proof of Lemma~\ref{Lemma:Domination}}
By definition of
$\widehat\beta$, $
\sqrt{\widehat Q(\hat\beta)} - \sqrt{\widehat Q (\beta_0)} \leq
\frac
{\lambda}{n} \|\Gamma\beta_0\|_{1} - \frac{\lambda}{n}\|\Gamma
\hat
\beta\|_{1}$. By convexity of $\sqrt{\widehat Q}$, by $-\widetilde S
\in\partial\sqrt{\widehat Q} (\beta_0)$, and by $\lambda/n \geq
cn\|
\Gamma^{-1}\widetilde S\|_{\infty}$, we have
$\sqrt{\hat Q (\hat\beta)} - \sqrt{\hat Q(\beta_0)}
\geq
-\widetilde S'\hat\delta
\geq
- \|\Gamma^{-1}\widetilde S\|_{\infty} \|\Gamma\hat\delta\|_{1}
\geq- \frac{\lambda}{cn} \|\Gamma\hat\delta\|_{1}
$ where $\hat\delta= \hat\beta-\beta_0$. Combining the lower and upper
bounds yields
$
\|\Gamma\hat\delta\|_{1} \geq c ( \|\Gamma(\beta_0 + \hat\delta)
\|
_1 - \|\Gamma\beta_0\|_{1} )$.
Thus, $\hat\delta\in R_c$; that $\hat\delta\in\Delta_{\bar c}$
follows by a standard argument based on elementary inequalities.
\end{pf*}

\begin{pf*}{Proof of Theorem~\ref{Thm:NewRateSquareRootLASSONonparametric}}
First, note that by Lemma~\ref{Lemma:Domination} we have $\widehat
\delta:=\widehat\beta-\beta_0 \in R_c$.
By optimality of $\hat\beta$ and definition of $\kk$, $\bbb=\lambda
\sqrt{s}/[n\kk]$, we have
%
%
\begin{equation}
\label{NewOPT} \sqrt{\widehat Q(\hat\beta)} - \sqrt{\widehat Q (
\beta_0)} \leq\frac
{\lambda}{n} \|\Gamma\beta_0
\|_{1} - \frac{\lambda}{n}\|\Gamma\hat\beta\|_{1} \leq\bbb
\|\widehat\delta\|_{2,n}.
\end{equation}
Multiplying both sides by $\sqrt{\widehat Q(\widehat\beta)} + \sqrt {\widehat Q(\beta_0)}$ and since $(a+b)(a-b)=a^2-b^2$
%
%
\begin{equation}
\label{Eq:Id}\|\widehat\delta\|_{2,n}^2 \leq2 \En \bigl[(
\sigma\varepsilon_\ii+r_\ii) x_\ii'
\widehat\delta \bigr] + \bigl(\sqrt{\widehat Q(\widehat\beta)} + \sqrt{\widehat
Q( \beta_0)} \bigr) \bbb\|\widehat\delta\|_{2,n}.
\end{equation}
From (\ref{NewOPT}), we have $
\sqrt{\widehat Q(\hat\beta)}
\leq\sqrt{\widehat Q (\beta_0)} + \bbb\|\widehat\delta\|_{2,n}$
so that
\[
\|\hat\delta\|_{2,n}^2 \leq 2\En \bigl[(\sigma
\varepsilon_\ii+r_\ii) x_\ii'
\widehat\delta \bigr] + 2\sqrt{\widehat Q(\beta_0)} \bbb\|\hat \delta
\| _{2,n} + \bbb^2\|\hat\delta\|_{2,n}^2.
\]
Since $|\En[(\sigma\varepsilon_\ii+r_\ii)x_\ii'\hat\delta
]|=\sqrt{\widehat
Q(\beta_0)}|\widetilde S'\hat\delta|\leq\sqrt{\widehat Q(\beta
_0)}\varrho_{c}\|\hat\delta\|_{2,n}$, we obtain
\begin{eqnarray*}
\|\hat\delta\|_{2,n}^2 & \leq& 2\sqrt{\widehat Q(
\beta_0)}\varrho_{c}\| \hat\delta\|_{2,n} + 2
\sqrt{\widehat Q(\beta_0)} \bbb\|\hat\delta\| _{2,n} +
\bbb^2\|\hat\delta\|_{2,n}^2,
\end{eqnarray*}
and the result follows provided $\bbb<1$.
\end{pf*}

\begin{pf*}{Proof of Theorem~\ref{Thm:Q}} We have $\hat\delta:=
\widehat
\beta- \beta_0 \in R_{c}$ under the condition that $\lambda/n \geq
c\|
\Gamma^{-1}\widetilde S\|_\infty$ by Lemma~\ref{Lemma:Domination}. We
also have
$\bbb=\lambda\sqrt{s}/[n\kk] < 1$ by assumption.

First, we establish the upper bound. By the previous proof, we have inequality
(\ref{NewOPT}). The bound follows from Theorem~\ref
{Thm:NewRateSquareRootLASSONonparametric} to bound $\|\hat\delta\|_{2,n}$.
To establish the lower bound, by convexity of $\sqrt{\widehat Q}$ and
the definition of $\varrho_{c}$ we have
$
\sqrt{\widehat Q(\widehat\beta)} - \sqrt{\widehat Q(\beta_0)} \geq
-\widetilde S'\hat\delta\geq-\varrho_{c}\|\hat\delta\|_{2,n}$.
Thus, by Theorem~\ref{Thm:NewRateSquareRootLASSONonparametric} we obtain
$\sqrt{\widehat Q(\widehat\beta)} - \sqrt{\widehat Q(\beta_0)}
\geq-2
\sqrt{\widehat Q(\beta_0)}\varrho_{c} \BB$.

Moreover, by the triangle inequality
\begin{eqnarray*}
\bigl\vert\sqrt{\widehat Q(\widehat\beta)}- \sigma \bigr\vert \leq
\bigl \vert\sqrt{\widehat
Q(\widehat\beta)}- \sigma \bigl\{\En \bigl[ \varepsilon_\ii^2
\bigr] \bigr\} ^{1/2} \bigr\vert + \sigma\bigl \vert \bigl\{\En \bigl[
\varepsilon_\ii^2 \bigr] \bigr\}^{1/2}-1 \bigr\vert
\end{eqnarray*}
and the right-hand side is bounded by $\|\widehat\beta-\beta_0\|_{2,n}
+ c_s + \sigma|\En[\varepsilon_\ii^2] - 1|$.
\end{pf*}

\begin{pf*}{Proof of Theorem~\ref{Thm:Sparsity}}
For notational convenience, we denote $\phi_n(m) = \phi_{\mathrm
{max}}(m,\Gamma ^{-1}\En[x_\ii x_\ii']\Gamma^{-1})$. 
We shall rely on the following lemma, whose proof is given after the
proof of this theorem.
%
%
\begin{lemma}[(Relating sparsity and prediction norm)]\label{Lemma:Proof-D}
Under Condition \ref{condASM}, let $G \subseteq\supp(\widehat\beta)$. For any
$\lambda> 0$, we have
\begin{eqnarray*}
\frac{\lambda}{n}\sqrt{\hat Q(\hat\beta)}\sqrt{ |G| } &\leq&\sqrt {|G|}\bigl\|
\Gamma^{-1}\widetilde S\bigr\|_\infty\sqrt{\widehat Q(
\beta_0)} \\
&&{}+ \sqrt{\phi_{\mathrm{max}}\bigl(|G\setminus T|,
\Gamma ^{-1}\En \bigl[x_\ii x_\ii'
\bigr]\Gamma^{-1}\bigr)} \| \hat\beta- \beta_0
\|_{2,n}.
\end{eqnarray*}
\end{lemma}
Define $\widehat m:= |\supp(\widehat\beta)\setminus T|$. In the event
$\lambda/n \geq c \|\Gamma^{-1}\widetilde S\|_\infty$, by Lemma~\ref
{Lemma:Proof-D}
%
%
\begin{equation}
\label{Eq:StepSparsity} \biggl(\sqrt{\frac{\hat Q(\hat
\beta
)}{\hat Q(\beta_0)}} - \frac{1}{c} \biggr)
\frac{\lambda}{n}\sqrt{\hat Q(\beta_0)}\sqrt{\bigl |\supp(\widehat \beta)\bigr|
} \leq\sqrt{\phi_n(\hat m)} \| \hat\beta- \beta_0
\|_{2,n}.
\end{equation}
Under the condition $\bbb= \lambda\sqrt{s}/[n \kk]<1$, we have by
Theorems \ref{Thm:NewRateSquareRootLASSONonparametric} and \ref
{Thm:Q} that
\[
\biggl(1 - 2\varrho_{c}\BB- \frac{1}{c} \biggr)
\frac{\lambda}{n}\sqrt{\hat Q(\beta_0)}\sqrt{ \bigl|\supp(\widehat \beta)\bigr|
} \leq\sqrt{\phi_n(\hat m)} 2\sqrt{\widehat Q(\beta_0)}
\BB,
\]
where $\BB= \frac{\varrho_{c} + \bbb}{1-\bbb^2}$. Since we assume $
2\varrho_{c} \BB\leq1/(c\cc)$, we have
\[
\sqrt{ \bigl|\supp(\widehat\beta)\bigr| } \leq2\cc\sqrt{\phi_n(\hat m) }
\frac
{n}{\lambda} \BB= \sqrt{s}\sqrt{\phi_n(\hat m) } 2\cc\BB/(\bbb
\kk),
\]
where the last equality follows from $\bbb= \lambda\sqrt{s}/[n\kk]$.

Let $L:= 2\cc\BB/(\bbb\kk)$. Consider any $m \in\mathcal{M}$, and
suppose $\widehat m > m$. Therefore, by the sublinearity of maximum
sparse eigenvalues (see Lemma~3 in \cite{BC-PostLASSO}), $\phi_n(\ell
m) \leq \lceil\ell  \rceil\phi_n(m)$ for $\ell\geq1$,
and $\widehat m
\leq
|\supp(\widehat\beta)|$ we have $\hat m \leq s \cdot \lceil
\frac {\hat m}{m}  \rceil\phi_n(m) L^2$. Thus, since $
\lceil k  \rceil< 2k$ for any $k\geq1$ we have
$ m < s \cdot2\phi_n(m) L^2$ which violates the condition of $m \in
\mathcal{M}$ and $s$. Therefore, we must have $\widehat m \leq m$.
Repeating the argument once more with $\widehat m \leq m$ we obtain
$ \hat m \leq s \cdot\phi_n(m) L^2$.
The result follows by minimizing the bound over $m \in\mathcal{M}$.

To show the second part, by Lemma~\ref{Lemma:NewRho} and $\lambda
/n\geq
c\|\Gamma^{-1}\widetilde S\|_\infty$, we have $\varrho_{c} \leq
\frac
{\lambda\sqrt{s}}{n\kappa_{\cc}}\frac{1+\cc}{c}$. Lemma~\ref
{Lemma:NewKappa} yields $\kk\geq\kappa_{\cc}$ and recall $\bbb=
\lambda\sqrt{s}/(n\kk)$. Therefore,
\begin{eqnarray*}
\BB/(\bbb\kk) &\leq&\frac{1+\{({\lambda\sqrt{s}}/{(n\kappa
_{\cc})})
({(1+\cc)}/{c})\}\{{n\kk}/{(\lambda\sqrt{s})}\}}{\kk(1-\bbb^2)}\\
&\leq& \frac
{1+{(1+\cc)}/{c}}{\kappa_{\cc}(1-\bbb^2)}=\frac{\cc}{\kappa
_{\cc
}(1-\bbb^2)}
\leq\frac{2\cc}{\kappa_{\cc}},
\end{eqnarray*}
where the last inequality follows from the condition $\bbb\leq1/\sqrt{2}$.
Thus, it follows that $4\cc^2 ( \BB/(\bbb\bar\kappa)
)^2\leq(4\cc^2/\kappa_{\cc})^2$ which implies $\mathcal
{M}^*\subseteq
\mathcal{M}$.
\end{pf*}

\begin{pf*}{Proof of Lemma~\ref{Lemma:Proof-D}}
Recall that $\Gamma= \diag(\gamma_1,\ldots,\gamma_p)$. $\hat\beta
$ is
the solution of a conic optimization problem (see Section H.1 of supplementary material
\cite{BCWsupp}). Let $\hat a$ denote the solution to its dual problem:
$\max_{a\in\RR^n} \En[y_\ii a_\ii] \dvtx\break \|\Gamma^{-1}\En
[x_{\ii
j}a_\ii]\|_\infty\leq\lambda/n, \|a\|\leq\sqrt{n}$. By strong duality
$ \En[ y_\ii\hat a_\ii] = \frac{\|Y - X\hat\beta\|}{\sqrt{n}} +\break
\frac
{\lambda}{n}\sum_{j=1}^p \gamma_j|\hat\beta_j|$.\vspace*{2pt}
Moreover, by the first-order optimality conditions,\break  $\En[x_{\ii j}\hat
a_\ii]\hat\beta_j =
\lambda\gamma_j|\hat\beta_j|/n$ holds for every $j=1,\ldots,p$.
Thus, we
have
\[
\En[ y_\ii\hat a_\ii] = \frac{\|Y - X\hat\beta\|}{\sqrt{n}} + \sum
_{j=1}^p \En[x_{\ii j}\hat
a_\ii]\hat\beta_j= \frac{\|Y - X\hat
\beta\|
}{\sqrt{n}} + \En \Biggl[\hat
a_\ii\sum_{j=1}^p
x_{\ii j}\hat\beta_j \Biggr].
\]
Rearranging the terms, we have $\En[ (y_\ii-x_\ii'\hat\beta) \hat
a_\ii] = \|Y - X\hat\beta\|/\sqrt{n}$.

If $\|Y-X\hat\beta\|=0$, we have $\sqrt{\hat Q(\hat\beta)} = 0$ and
the statement of the lemma trivially holds. If $\|Y-X\hat\beta\|>0$,
since $\|\hat a\|\leq
\sqrt{n}$ the equality can only hold for $\hat a =
\sqrt{n}(Y-X\hat\beta)/\|Y-X\hat\beta\|=(Y-X\hat\beta)/\sqrt {\hat
Q(\hat\beta)}$.\vadjust{\eject}

Next, note that for any $j \in\supp(\widehat\beta)$ we have $
\En[ x_{\ii j}\hat a_\ii] = \sign(\hat\beta_j)\lambda\gamma_j/n$.
Therefore, for any subset $G \subseteq\supp(\widehat\beta)$ we have
\begin{eqnarray*}
&& \sqrt{\hat Q(\hat\beta)} \sqrt{|G|}\lambda\\
&&\qquad=
\bigl\| \Gamma^{-1} \bigl(X'(Y - X\hat\beta)
\bigr)_{G} \bigr\|
\\
&&\qquad \leq\bigl\|\Gamma^{-1} \bigl(X'(Y - X
\beta_0) \bigr)_{G} \bigr\| +\bigl\|\Gamma^{-1}
\bigl(X'X( \beta_0-\hat\beta) \bigr)_{G} \bigr\|
\\
&&\qquad \leq\sqrt{|G|} n\bigl\|\Gamma^{-1}\En \bigl[x_\ii(\sigma
\varepsilon_\ii+r_\ii) \bigr]\bigr\|_{\infty}\\
&&\qquad\quad{} + n \sqrt
{\phi_{\mathrm{max}}\bigl(|G\setminus T|,\Gamma^{-1}\En
\bigl[x_\ii x_\ii' \bigr]
\Gamma^{-1}\bigr)} \|\hat\beta- \beta_0 \|_{2,n}
\\
&&\qquad = \sqrt{|G|} n\sqrt{\widehat Q(\beta_0)}\bigl\|\Gamma^{-1}
\widetilde S\bigr\|_{\infty} \\
&&\qquad\quad{}+ n \sqrt{\phi_{\mathrm{max}}\bigl(|G
\setminus T|, \Gamma^{-1}\En \bigl[x_\ii
x_\ii' \bigr]\Gamma^{-1}\bigr)} \|\hat \beta-
\beta_0 \|_{2,n},
\end{eqnarray*}
where we used
\begin{eqnarray*}  \bigl\|\Gamma^{-1} \bigl(X'X(
\hat\beta-\beta_0) \bigr)_{G}\bigr \| &\leq&\sup
_{\|\alpha_{T^c}\|_0\leq|G\setminus T|, \|\alpha\|\leq1}\bigl| \alpha '\Gamma^{-1}X'X(
\hat\beta-\beta_0)\bigr|
\\
&\leq&\sup_{\|\alpha_{T^c}\|_0\leq|G\setminus T|, \|\alpha\|\leq
1}\bigl\| \alpha'
\Gamma^{-1}X'\bigr\|\bigl\|X(\hat\beta- \beta_0)\bigr\|
\\
&\leq & n\sqrt{\phi_{\mathrm{max}}\bigl(|G
\setminus T|, \Gamma^{-1}\En \bigl[x_\ii
x_\ii' \bigr]\Gamma^{-1}\bigr)}\| \hat\beta-
\beta_0\|_{2,n}.
\end{eqnarray*}
\upqed\end{pf*}

\begin{pf*}{Proof of Theorem~\ref{Thm:2StepNonparametric}}
In this
proof, let
$f = (f_1,\ldots,f_n)'$, $R = (r_1,\ldots,\break  r_n)'$, $\varepsilon=
(\varepsilon
_1,\ldots,\varepsilon_n)'$ ($n$-vectors) and $X = [x_1;\ldots;x_n]'$ (an
$n\times p$ matrix).
For a set of indices $S \subset\{1,\ldots,p\}$, define $\mathcal
{P}_{S} = X[S](X[S]'X[S])^{-} X[S]'$,
where we interpret $\mathcal{P}_{S}$ as a null operator if $S$ is
empty. We have that $f - X\widetilde\beta= ( I - \mathcal{P}_{\hat
T})f - \sigma\mathcal{P}_{\hat T}\varepsilon$, where $I$ is the identity
operator. Therefore,
%
%
\begin{eqnarray}
\label{eqPL} %
\sqrt{n}\|
\beta_0 - \widetilde\beta\|_{2,n} &=&\| X\beta_0
- X\widetilde\beta\| \nonumber\\
&=&| f - X\widetilde\beta- R \|
\nonumber
\\[-8pt]
\\[-8pt]
\nonumber
& =&\bigl \| (I-\mathcal{P}_{\hat T})f - \sigma\mathcal{P}_{\hat
T}
\varepsilon- R \bigr\| \\
&\leq&\bigl\|(I- \mathcal{P}_{\hat T})f\bigr\| +\sigma\|
\mathcal{P}_{\hat T} \varepsilon\|+ \|R\| \nonumber
\end{eqnarray}
where $\|R\| \leq\sqrt{n} c_s$.
Since for $\widehat m = |\hat T\setminus T |$, we have
\[
\bigl\| X[\widehat T] \bigl(X[\widehat T]'X[\widehat T]
\bigr)^{-}\bigr\|_{\mathrm{op}} \leq\sqrt{1/\phi_{\mathrm{min}}
\bigl(\widehat m, \En \bigl[x_\ii x_\ii'
\bigr] \bigr)}=\sqrt{1/ \phi_{\mathrm{min}}(\widehat m)},
\]
[where the bound is interpreted as $+\infty$ if $\phi_{\mathrm
{min}}(\widehat m) =
0$], the term $\|\mathcal{P}_{\hat T }\varepsilon\|$ in (\ref{eqPL})
satisfies
\[ \|\mathcal{P}_{\hat T}\varepsilon\|  \leq
\sqrt{1/\phi_{\mathrm{min}}(\widehat m)} \bigl\|X[\hat T]'\varepsilon /
\sqrt{n}\bigr\| \leq \sqrt{|\widehat T|/\phi_{\mathrm{min}}(\widehat m)}
\bigl\|X'\varepsilon /\sqrt{n} \bigr\|_\infty.
\]

Therefore, we have
$
\|\widetilde\beta- \beta_0\|_{2,n} \leq\frac{\sigma\sqrt{s+\hat m
}\|\En[x_\ii\varepsilon_\ii]\|_\infty}{\sqrt{\phi_{\mathrm
{min}}(\hat m)}} +
c_s +
c_{\widehat T}$, where $c_{\widehat T} = \min_{\beta\in\RR^p} \sqrt {\En
[(f_\ii-x_\ii'\beta_{\hat T})^2]}$. Since $\supp(\widehat\beta
)\subseteq\widehat T$ and (\ref{MainFSCondition}) holds,
\begin{eqnarray*}c_{\widehat T} &=& \min_{\beta\in\RR^p}
\bigl\{\En \bigl[ \bigl(f_\ii-x_\ii'
\beta_{\hat T} \bigr)^2 \bigr] \bigr\}^{1/2} \leq
\bigl\{ \En \bigl[ \bigl(f_\ii-x_\ii'
\widehat\beta \bigr)^2 \bigr] \bigr\}^{1/2}
\\
& \leq& c_s + \|\beta_0-\widehat\beta\|_{2,n}
\leq c_s + 2\sqrt{\hat Q(\beta_0)}\BB,
\end{eqnarray*}
where we have used Theorem~\ref{Thm:NewRateSquareRootLASSONonparametric}.
\end{pf*}

\begin{pf*}{Proof of Theorem~\ref{Thm:Perfect}}
Note that because
$\sigma= 0$ and $c_s=0$, we have $\sqrt{\widehat Q(\beta_0)} = 0$ and
$\sqrt{\widehat Q(\hat\beta)} = \|\hat\beta-\beta_0\|_{2,n}$.
Thus, by
optimality of $\hat\beta$ we have
$ \|\hat\beta-\beta_0\|_{2,n} + \frac{\lambda}{n}\|\Gamma\hat
\beta\|
_1\leq\frac{\lambda}{n}\|\Gamma\beta_0\|_1$ which implies\vadjust{\goodbreak} $\|
\Gamma
\hat\beta\|_1\leq\|\Gamma\beta_0\|_1$. Moreover, $\delta= \hat
\beta
-\beta_0$ satisfies
$\| \delta\|_{2, n} \leq\frac{\lambda}{n} ( \|\Gamma\beta_0\|_1-
\|
\Gamma\hat\beta\|_1) \leq\bbb\| \delta\|_{2,n} $,
where $\bbb= \frac{\lambda\sqrt{s}}{n\kk} < 1$. Hence, $\|\delta\|
_{2,n} =
0$.

Since $\|\Gamma\hat\beta\|_1\leq\|\Gamma\beta_0\|_1$ implies
$\delta
\in\Delta_1$, it follows
that $0 = \sqrt{s}\|\delta\|_{2,n} \geq\|\Gamma\delta_T\|_1/\kappa
_1\geq\frac{1}{2}\|\Gamma\delta\|_1/\kappa_1$, which implies that
$\delta=0$ if $\kappa_1 >0$.
\end{pf*}

\begin{pf*}{Proof of Theorem~\ref{Thm:Symmetric}}
If $\lambda/n\geq c\|\Gamma^{-1}\widetilde S\|_\infty$ and $\bbb=
\lambda\sqrt{s}/[n\kk] < 1$, by Theorem~\ref
{Thm:NewRateSquareRootLASSONonparametric} we have
$\|\hat\beta-\beta_0\|_{2,n} \leq2\sqrt{\widehat Q(\beta_0)}\BB$,
and the bound on the prediction norm follows by
$ \sqrt{\widehat Q(\beta_0)} \leq c_s + \sigma\sqrt{\En
[\varepsilon
_\ii^2]}$.

Thus, we need to show that the choice of $\lambda$ and $\Gamma$ ensures
the event $\lambda/n\geq c\|\Gamma^{-1}\widetilde S\|_\infty$ with
probability no less than $1-\alpha-\eta_1-\eta_2$. Since $\gamma
_j=\max_{1\leq i \leq n} |x_{ij}|\geq\En[x_{\ii j}^2] = 1$, by the
choice of
$u_n$ we have
\begin{eqnarray*} P \biggl( c\bigl\|\Gamma^{-1}\widetilde S
\bigr\|_\infty> \frac{\lambda}{n} \biggr) &\leq& P \biggl( { \max
_{1\leq j\leq p}} \frac{ c |\En[(\sigma
\varepsilon_\ii+r_\ii)x_{\ii j}]|}{\gamma_j\sqrt{\En[(\sigma
\varepsilon_\ii
)^2]}} > \frac{\lambda}{n(1+u_n)^{1/2}} \biggr) +
\eta_1 \\
&\leq& I + \mathit{II} + \eta_1,
\\
 I&:=& P \biggl( { \max_{1\leq j\leq p}} \frac{ |\En[\varepsilon
_\ii
x_{\ii
j}]|}{\gamma_j\sqrt{\En[\varepsilon_\ii^2]}} >
\frac{\sqrt{2\log
(2p/\alpha
)}}{\sqrt{n}} \biggr),\\
 \mathit{II}&:=& P \biggl( \frac{ \|\En[r_\ii x_{\ii
}]\|
_\infty}{\sqrt{\En[(\sigma\varepsilon_\ii)^2]}} >
\frac
{(1+u_n)^{1/2}}{\sqrt{n}} \biggr).
\end{eqnarray*}
We invoke the following lemma, which is proven in \cite{BC-PostLASSO}---see step 2 of the proof of~\cite{BC-PostLASSO}'s
Theorem~2; for completeness, supplementary material \cite{BCWsupp}
also provides the proof.
%
%
\begin{lemma}\label{Lemma:InfNormApprox}
Under Condition \ref{condASM}, we have $ \| \En[ x_\ii r_\ii]\|_\infty\leq
\min
\{\frac{\sigma}{\sqrt{n}}, c_s \}$.
\end{lemma}
By Lemma~\ref{Lemma:InfNormApprox}, $\|\En[r_\ii x_\ii]\|_\infty
\leq
\sigma/\sqrt{n}$ and $P(\En[\varepsilon_\ii^2] \leq\{1+u_n\}^{-1})
\leq
\eta_2$, we have $\mathit{II}\leq P ( \sqrt{\En[(\varepsilon_\ii)^2]}\leq
\{
1+u_n\}^{-1/2} ) \leq\eta_2$. Also,
\[
I \leq p \max_{1\leq j\leq p} P \biggl( \frac{ \sqrt{n}|\En
[\varepsilon
_\ii
x_{\ii j}]|}{\sqrt{\En[x_{\ii j}^2\varepsilon_\ii^2]}} > \sqrt {2
\log (2p/\alpha)} \biggr) \leq\alpha
\]
where we used that $\gamma_j\sqrt{\En[\varepsilon_\ii^2]} \geq
\sqrt{\En
[x_{\ii j}^2\varepsilon_\ii^2]}$,
the union bound, and the subGaussian inequality for self-normalized
sums stated in Theorem~2.15 of \cite{delapena}, since $\varepsilon_i$'s
are independent and symmetric by assumption.\vspace*{-2pt}
\end{pf*}

\begin{pf*}{Proof of Corollary~\ref{Thm:t(2)}}
See supplementary material \cite{BCWsupp}.\vspace*{-2pt}
\end{pf*}

\section{Proofs of Section~\texorpdfstring{\lowercase{\protect\ref{cec:primitive}}}{4}}
\vspace*{-2pt}
\begin{pf*}{Proof of Theorem~\ref{CorFinite:SqrtLASSO}}
The proof is given in supplementary material \cite{BCWsupp} and follows from Theorems \ref{Thm:NewRateSquareRootLASSONonparametric}--\ref{Thm:Sparsity}
with the help of Lemma~7 in supplementary
material~\cite{BCWsupp}.\vspace*{-2pt}
\end{pf*}

\begin{pf*}{Proof of Corollary~\ref{Thm:SigmaInference}} See supplementary material \cite
{BCWsupp}.\vspace*{-2pt}
\end{pf*}

\section{Proofs for Section~\texorpdfstring{\lowercase{\protect\ref{sec5}}}{5}}
\vspace*{-2pt}
\begin{pf*}{Proof of Theorem~\ref{theorem:semiparametric}} Throughout
the proof,
we use the notation
\[
B(w): = \max_{j, k} \sup_{\nu\in\Theta\times T} \bigl|
\partial_{\nu_k} \psi_j(w, \nu)\bigr|,\qquad \tau_n:=
\sqrt{ s \log(p \vee n)/n}.
\]

\textit{Step} 1. (A preliminary rate result.) In this step, we claim that $\|
\hat\theta- \theta_0\| \lesssim_P \tau_n$.
By definition, $\| \En[\psi(w, \hat\theta, \hat h(z))]\| \leq
\varepsilon
_n$ and $\hat\theta\in\Theta$ with probability $1- o(1)$, which implies
via triangle inequality that with the same probability:
\[
\bigl\| \barEp \bigl[\psi \bigl(w, \theta, h_0(z) \bigr) \bigr]
\vert_{\theta= \hat
\theta} \bigr\| \leq\varepsilon_n + I_1 +
I_2 \lesssim_P \tau_n,
\]
where $I_1$ and $I_2$ are defined in step 2 below, and the last bound
also follows from step 2 below and from the numerical tolerance obeying
$\varepsilon_n = o(n^{-1/2})$ by assumption. Since by Condition \ref{condSP}(iv),
$2^{-1} ( \| J_n(\hat\theta- \theta_0)\| \wedge c)$ is weakly smaller
than the left-hand side of the display, we conclude that
$\| \hat\theta- \theta_0\| \lesssim_P \tau_n$, using that singular
values of $J_n$ are bounded away from zero uniformly in $n$ by
Condition \ref{condSP}(v).

\textit{Step} 2. (Define and bound $I_1$ and $I_2$.) We claim that:
\begin{eqnarray*}
I_1 &:=& \sup_{\theta\in\Theta} \bigl\llVert\En\psi
\bigl(w, \theta, \hat h(z) \bigr) - \En\psi \bigl(w, \theta, h_0(z)
\bigr) \bigr\rrVert\lesssim_P \tau_n,
\\
I_2 &:= & \sup_{\theta\in\Theta} \bigl\llVert\En\psi
\bigl(w, \theta, h_0(z) \bigr) - \barEp\psi \bigl(w, \theta,
h_0(z) \bigr) \bigr\rrVert\lesssim_P n^{-1/2}.
\end{eqnarray*}

Using Taylor's expansion, for $\tilde h(z;\theta,j)$ denoting a point
on a line connecting vectors $h_0(z)$ and $h(z)$, which can depend on
$\theta$ and $j$,
\begin{eqnarray*}
I_1 & \leq& \sum_{j=1}^d
\sum_{m=1}^M \sup_{\theta\in\Theta}
\bigl\vert\En \bigl[ \partial_{t_m} \psi_j \bigl(w, \theta,
\tilde h(z;\theta,j) \bigr) \bigl(\hat h_m(z) - h_{0m}(z)
\bigr) \bigr]\bigr \vert
\\
& \leq& d M \bigl\{\En B^2(w) \bigr\}^{1/2} \max
_{m} \bigl\{ \En \bigl(\hat h_m(z) -
h_{0m}(z) \bigr)^2 \bigr\}^{1/2},
\end{eqnarray*}
where the last inequality holds by definition of $B(w)$ given earlier
and H\"older's inequality.
Since $\barEp B^2(w) \leq C$ by Condition \ref{condSP}(ii), $\En B^2(w)\lesssim_P
1$ by Markov's inequality. By this, by Condition \ref{condAS}(i),
by $d$ and $M$ fixed, conclude that $I_1 \lesssim_P \tau_n$.

Using Jain--Marcus' theorem, as stated in Example~2.11.13 in \cite
{vdV-W}, we conclude that $\sqrt{n} I_2 \lesssim_P 1$.
Indeed the hypotheses of that example follow from the assumption that
$\Theta$ is a fixed compact subset of $\RR^d$,
and from the Lipschitz property, $ \| \psi(w, \theta, h_0(z)) - \psi(w,
\tilde\theta, h_0(z))\| \leq\sqrt{d} B(w) \| \tilde\theta- \theta
\|$
holding uniformly for all $\theta$ and $\tilde\theta$ in $\Theta$,
with $\barEp B^2(w) \leq C$.

\textit{Step} 3. (Main step.) We have that $\sqrt{n} \| \En\psi(w, \hat\theta,
\hat h(z) ) \| \leq\varepsilon_n \sqrt{n}$.
Application of Taylor's theorem and the triangle inequality gives
\[
\bigl\llVert\sqrt{n} \En\psi \bigl(w, \theta_0, h_0(z)
\bigr) + J_n \sqrt{n} (\hat\theta- \theta_0) \bigr
\rrVert\leq\varepsilon\sqrt{n} + \|\mathit{II}_1\| + \| \mathit{II}_2\| +
\|\mathit{II}_3\| = o_P(1),
\]
where $J_n = \barEp\partial_\theta\psi(w, \theta_0, h_0(z))$, the
terms $\mathit{II}_1$, $\mathit{II}_2$ and $\mathit{II}_3$ are defined and bounded below in step
4; the $o_P(1)$ bound
follows from step 4 and from $\varepsilon_n \sqrt{n} = o(1)$ holding by
assumption. Conclude using
Condition \ref{condSP}(iv) that
\[
\bigl\llVert J_n^{-1} \sqrt{n} \En\psi \bigl(w,
\theta_0, h_0(z) \bigr) + \sqrt{n} (\hat\theta-
\theta_0) \bigr\rrVert= o_P(1),
\]
which verifies the first claim of the theorem. Application of
Liapunov's central limit theorem in conjunction
with Condition \ref{condSP}(v) and the conditions on $\Omega_n$ imposed by the
theorem imply the second claim.

\textit{Step} 4. (Define and bound $\mathit{II}_1$, $\mathit{II}_2$ and $\mathit{II}_3$.) Let $\mathit{II}_1:=
(\mathit{II}_{1j})_{j=1}^d$ and $\mathit{II}_2 = (\mathit{II}_{2j})_{j=1}^d$, where
\begin{eqnarray*}
 \mathit{II}_{1j}&:=& \sum_{m = 1}^M
\sqrt{n} \En \bigl[ \partial_{t_m} \psi_j \bigl(w,
\theta_0, h_0(z) \bigr) \bigl(\hat h_m(z) -
h_{0m}(z) \bigr) \bigr],
\\
 \mathit{II}_{2j}&:=&\sum_{r,k = 1}^K
\sqrt{n} \En \bigl[ \partial_{\nu_k} \partial_{\nu_r}
\psi_j \bigl(w, \tilde\nu(w;j) \bigr) \bigl\{ \hat
\nu_r(w) - \nu_{0r}(w) \bigr\} \bigl\{ \hat
\nu_k(w) - \nu_{0k}(w) \bigr\} \bigr],
\\
 \mathit{II}_3&:=& \sqrt{n} \bigl( \En\partial_\theta\psi \bigl(w,
\theta_0, h_0(z) \bigr) - J_n \bigr) ( \hat
\theta- \theta_0),
\end{eqnarray*}
where $\nu_0(w): = ( \nu_{0k}(w))_{k=1}^K: = (\theta_0', h_0(z)')'$;
$K= d+ M$; $\hat\nu(w): =\break  ( \hat\nu_k(w))_{k=1}^K:=
(\hat\theta', \hat h(z)')'$, and $\tilde\nu(w;j)$ is a vector on the
line connecting $\nu_0(w)$ and $\hat\nu(w)$ that may depend on $j$. We
show in this step that $\|\mathit{II}_1\| + \|\mathit{II}_2\| + \|\mathit{II}_3\| \lesssim_P o(1)$.

The key portion of the proof is bounding $\mathit{II}_{1}$, which is very
similar to the argument
given in \cite{BellChenChernHans:nonGauss} (pages~2421--2423). We repeat
it here for completeness.
We split $\mathit{II}_{1} = \mathit{III}_{1} + \mathit{III}_{2} = (\mathit{III}_{1j})_{j=1}^d +
(\mathit{III}_{2j})_{j=1}^d$, where
\begin{eqnarray*}
\mathit{III}_{1j} &:= & \sum_{m=1}^M
\sqrt{n} \En \Biggl[ \partial_{t_m} \psi_j \bigl(w,
\theta_0, h_0(z) \bigr) \sum
_{l=1}^p P_{ml}(z) ( \hat
\beta_{ml} - \beta_{0ml}) \Biggr],
\\
\mathit{III}_{2j} &:= & \sum_{m=1}^M
\sqrt{n} \En \bigl[ \partial_{t_m} \psi_j \bigl(w,
\theta_0, h_0(z) \bigr) r_m(z) \bigr].
\end{eqnarray*}
Using H\"older inequality, $\max_j |\mathit{III}_{1j}| \leq M \max_{j,m,l} |
\sqrt{n} \En\dot\psi_{mjl}(w) | \| \hat\beta_m -\break  \beta_{0m} \|_1$.
By Condition \ref{condAS}(i) $\max_m \|\hat\beta_m - \beta_{0m}\|_1 \leq C
\sqrt{s} \tau_n$ with probability at least $1-\delta_n$.
Moreover, using $\Ep\dot\psi_{mjl} (w_i) = 0$ for all $i$, which
holds by the orthogonality property (\ref{eq:orthogonality}),
and that $\max_{j, m, l } \En| \dot\psi_{mjl}(w) |^2 \leq L^2_n$ with
probability at least $1-\delta_n$ by Condition \ref{condAS}(ii), we can apply
Lemma~\ref{Lemma: MDSN} on the moderate deviations for self-normalized
sum, following the idea in \cite{BellChenChernHans:nonGauss}, to
conclude that $\max_{j,m,l} | \sqrt{n} \En\dot\psi_{mjl}(w) | \leq
\sqrt{2 \log(pn)} L_n$ with probability $1- o(1)$. Note that this
application requires
the side condition $\sqrt{2 \log(pn)} M_n n^{-1/6} = o(1)$ be
satisfied for $M_n$ defined in Condition \ref{condAS}(ii), which indeed holds by
Condition \ref{condAS}(iii).
We now recall the details of this calculation:
\begin{eqnarray*}
&& P \Bigl( \max_{j,m,l}\bigl | \sqrt{n} \En\dot
\psi_{mjl}(w) \bigr| > \sqrt{2 \log(pn)} L_n \Bigr)
\\
&&\qquad \leq P \Bigl( \max_{j,m,l} \bigl| \sqrt{n} \En\dot
\psi_{mjl}(w) \bigr|/\sqrt{\En\bigl|\dot\psi_{mjl}(w) \bigr|^2 }
> \sqrt{2 \log(pn)} \Bigr) + \delta_n
\\
&& \qquad\leq d M p \max_{j,m,l} P \bigl(\bigl| \sqrt{n} \En\dot
\psi_{mjl}(w) \bigr|/\sqrt{\En\bigl|\dot\psi_{mjl}(w) \bigr|^2 }
> \sqrt{2 \log(pn)} \bigr) + \delta_n
\\
& &\qquad \leq d M p 2 \bigl(1- \Phi \bigl( \sqrt{2 \log(pn)} \bigr) \bigr) \bigl(1 +
o(1) \bigr) + \delta_n \leq d M p \frac{2}{pn} \bigl(1 + o(1)
\bigr) + \delta_n\\
&&\qquad = o(1),
\end{eqnarray*}
where the penultimate inequality occurs due to the application of
Lemma~\ref{Lemma: MDSN} on moderate deviations for self-normalized
sums. Putting bounds together we conclude that $\|\mathit{III}_{1}\| \leq\sqrt {d} \max_{j} |\mathit{III}_{1j}| \lesssim_P L_n \sqrt{ \log(p \vee n)} \sqrt{s}
\tau_n = o(1)$, where $o(1)$ holds by the growth restrictions imposed
in Condition \ref{condAS}(iii).

The bound on $\mathit{III}_2$ also follows similarly to \cite
{BellChenChernHans:nonGauss}. $\mathit{III}_{2j}$ is a sum of $M$ terms, each
having mean zero and variance of order $s \log(p \vee n)/n = o(1)$.
Indeed, the mean zero occurs because
\[
n^{-1/2} \sum_{i=1}^n \Ep \bigl[
\partial_{t_m} \psi_j \bigl(w_i,
\theta_0, h_0(z_i) \bigr)
r_m(z_i) \bigr] = n^{-1/2} \sum
_{i=1}^n \Ep \bigl[ 0 \cdot r_m(z_i)
\bigr] =0
\]
for each $m$th term, which holds by $\Ep[ \partial_{t_m} \psi_j(w_i,
\theta_0, h_0(z_i)) | z_i] =0$, that is, the orthogonality property
(\ref{eq:orthogonality}), and the law of iterated expectations. To
derive the variance bound, note that for each $m$th term the variance is
\[
n^{-1} \sum_{i=1}^n \Ep \bigl[
\bigl\{\partial_{t_m} \psi_j \bigl(w_i,
\theta_0, h_0(z_i) \bigr) \bigr
\}^2 r_m^2(z_i) \bigr] \leq C
\barEp \bigl[r^2_m(z) \bigr] \leq C^2 s
\log(p \vee n)/n,
\]
which holds by $ \Ep[ \{\partial_{t_m} \psi_j(w_i, \theta_0,
h_0(z_i))\}^2|z_i] \leq
\Ep[ B^2(w) |z_i] \leq C$ a.s. by virtue of Condition \ref{condSP}(iii), and
the law iterated expectations; the last bound in the display holds by
\ref{condAS}(i). Hence, $\operatorname{var}(\mathit{III}_{2j})\leq M^2 C^2 s \log(p
\vee
n)/n\lesssim s \log(p \vee n)/n = o(1)$. Therefore, $\|
\mathit{III}_{2}\| \leq\sum_{j=1}^d | \mathit{III}_{2j}| \lesssim_P \sqrt{s \log(p
\vee
n)/n} = o(1)$ by Chebyshev's inequality.

To deduce that $\|\mathit{II}_2\| = o_P(1)$, we use Condition \ref{condAS}(i)--(iii), the
claim of step~1, and H\"older inequalities, concluding that
\[
\max_j |\mathit{II}_{2j}| \leq\sqrt{n} K^2
L_n \max_k \En \bigl\{ \hat
\nu_k(w) - \nu_{0k}(w) \bigr\}^2
\lesssim_P \sqrt{n} L_n \tau_n^2
= o(1).
\]

Finally, since $\|\mathit{II}_3\| \leq\sqrt{n} \| ( \En\partial_\theta
\psi(w, \theta_0, h_0(z)) - J_n ) \|_{\mathrm{op}} \| \hat\theta-
\theta
_0\|$ and since $\| \En\partial_\theta\psi(w, \theta_0, h_0(z)) -
J_n \|_{\mathrm{op}} \lesssim_P n^{-1/2}$ by
Chebyshev's inequality, using that $\barEp B^2(w) \leq C$ by Condition
\ref{condAS}(ii), and $\| \hat\theta- \theta_0\| \lesssim_P \tau_n $ by step 1,
conclude that $\|\mathit{II}_3\| \lesssim_P \tau_n =o(1)$.
\end{pf*}
\end{appendix}




\section*{Acknowledgements}
We are grateful to the Editors and two referees for thoughtful comments
and suggestions, which helped improve the paper substantially. We also
thank seminar participants at the 2011 Joint Statistical Meetings, 2011
INFORMS, Duke and MIT for many useful suggestions. We gratefully
acknowledge research support from the NSF.



\begin{supplement}[id=suppA]
\stitle{Supplementary material}
\slink[doi]{10.1214/14-AOS1204SUPP} 
\sdatatype{.pdf}
\sfilename{aos1204\_supp.pdf}
\sdescription{The material
contains deferred proofs, additional theoretical results on convergence rates
in $\ell_2, \ell_1$ and $\ell_\infty$, lower bound on the prediction
rate, and Monte-Carlo simulations.}
\end{supplement}



%

%

\printaddresses

\end{document}